\documentclass[times,sort&compress,3p]{elsarticle}
\journal{Journal of Multivariate Analysis}
\usepackage[labelfont=bf]{caption}

\usepackage{amsmath,amsfonts,amssymb,amsthm,booktabs,color,epsfig,graphicx,url}
\let\oldproofname=\proofname
\renewcommand{\proofname}{\rm\bf{\oldproofname}}

\usepackage[colorlinks, citecolor={blue}]{hyperref}

\usepackage{amsfonts,amscd,amssymb}
\usepackage{amsthm,amsmath,natbib}
\usepackage{algorithm,algorithmicx,algpseudocode}
\usepackage{bm}
\usepackage{bbm} 
\usepackage[table]{xcolor}
\usepackage{verbatim}
\usepackage{graphicx}
\usepackage{setspace}
\usepackage{natbib}

\usepackage{enumitem}
\usepackage{listings}
\usepackage[textsize=tiny]{todonotes}
\usepackage{tikz}
\usetikzlibrary{shapes.misc}
\usepackage{etoolbox}
\usepackage{appendix}
\usepackage{subcaption}
\usepackage{wrapfig}
\usepackage{xr}
\usepackage{booktabs}
\usepackage{multirow}

\usepackage{mathbbol}
\usepackage{thmtools}

\usetikzlibrary{matrix}
\usetikzlibrary{backgrounds}
\usetikzlibrary{calc}
\usetikzlibrary{arrows,shapes}
\usetikzlibrary{decorations}
\usetikzlibrary{decorations.pathmorphing}
\usetikzlibrary{fit}
\usetikzlibrary{decorations.pathreplacing}
\usetikzlibrary{shapes.misc}
\usetikzlibrary{shapes.geometric}

\newtoggle{quickdraw}

\definecolor{lightgrey}{rgb}{0.9,0.9,0.9}
\definecolor{darkgreen}{rgb}{0,0.3,0}

\definecolorset{rgb}{}{}{darkred,0.8,0,0;darkgreen,0,0.5,0;darkblue,0,0,0.5}

\newtheorem{thm}{Theorem}
\newtheorem{lemma}{Lemma}
\newtheorem{prop}{Proposition}
\newtheorem{cor}{Corollary}

		\newcommand{\order}[1]{{\cal O}\hspace{-0.2em}\left( #1 \right)}

		\definecolor{trevorblue}{rgb}{0.330, 0.484, 0.828}
		\definecolor{trevoryellow}{rgb}{0.829, 0.680, 0.306}

\graphicspath{{figures/}}

\begin{document}

\begin{frontmatter}

\title{Generating MCMC proposals by randomly rotating the regular simplex}

\author{Andrew J.~Holbrook\corref{mycorrespondingauthor}}

\address{Department of Biostatistics, University of California, Los Angeles, CA 90095, USA}

\cortext[mycorrespondingauthor]{Email address: \url{aholbroo@g.ucla.edu}}

\begin{abstract}
We present the \emph{simplicial sampler}, a class of parallel MCMC methods that generate and choose from multiple proposals at each iteration.  The algorithm's multiproposal randomly rotates a simplex connected to the current Markov chain state in a way that inherently preserves symmetry between proposals. As a result, the simplicial sampler leads to a simplified acceptance step: it simply chooses from among the simplex nodes with probability proportional to their target density values. We also investigate a multivariate Gaussian-based symmetric multiproposal mechanism and prove that it also enjoys the same simplified acceptance step. This insight leads to significant theoretical and practical speedups.  While both algorithms enjoy natural parallelizability, we show that conventional implementations are sufficient to confer efficiency gains across an array of dimensions and a number of target distributions. 
\end{abstract}


\end{frontmatter}

\section{Introduction}\label{sec:intro}

\newcommand{\x}{\mathbf{x}}
\renewcommand{\v}{\mathbf{v}}
\newcommand{\TTheta}{\boldsymbol{\theta}}
\newcommand{\TTTheta}{\boldsymbol{\Theta}}
\newcommand{\tr}{\mbox{tr}}
\newcommand{\X}{\mathbf{X}}
\renewcommand{\u}{\mathbf{u}}
\newcommand{\QQ}{\mathbf{Q}}
\newcommand{\haar}{\mathcal{H}}
\newcommand{\orthog}{\mathcal{O}}
\newcommand{\SSigma}{\boldsymbol{\Sigma}}

Generating random samples from probability distributions is a crucial task in many of the quantitative sciences. Markov chain Monte Carlo (MCMC) algorithms sample from target distributions by constructing Markov chains that maintain their targets as equilibrium distributions. The classic Metropolis-Hastings (MH) algorithm \citep{metropolis1953equation,hastings1970monte} builds such a Markov chain by iteratively generating a single proposal state conditioned on the current state and randomly accepting or rejecting that proposal based on a ratio involving the target and proposal densities evaluated at the current and proposed states.  The majority of MCMC techniques share this general structure, randomly proposing and accepting a single proposal at each iteration, but some relatively recent methods make use of multiple proposals (or, multiproposals) at each step with the goals of, e.g., reducing autocorrelation between Markov chain states or making better use of parallel computing resources.

\citet{liu2000multiple} develops the first such multiproposal technique, the multiple-try Metropolis algorithm (MTM). Each iteration of MTM starts by proposing multiple states $\TTheta_p^*$, $p\in\{1,\dots,P\}$ from a proposal distribution $q(\TTheta^{(s)},\cdot)$ and choosing a candidate $\TTheta^*$ from among them with probability proportional to their respective target density values $\pi(\TTheta^*_p)$.  After this, the MTM algorithm generates an additional $P-1$ states $\TTheta_1',\dots,\TTheta_{P-1}'$ from the distribution $q(\TTheta^{*},\cdot)$.  In the case that $q(\cdot,\cdot)$ is symmetric, the MTM algorithm then accepts $\TTheta^*$ with probability
\begin{align}\label{eq:mtmMetropolis}
	1   \wedge \frac{\pi(\TTheta^*_1) + \cdots + \pi(\TTheta^*_P) }{\pi(\TTheta'_1) + \cdots + \pi(\TTheta'_{P-1}) + \pi(\TTheta^{(s)} ) } \, .
\end{align}
Whereas MTM is the most widely used multiproposal MCMC algorithm, it has certain drawbacks. First, to use $P$ proposals, one must generate $O(2P)$ states and evaluate the target density $\pi(\cdot)$ $O(2P)$ times. Second, after randomly selecting a potentially good candidate state $\TTheta^*$, one might nonetheless reject it based on the target density values of the other proposed states within the Metropolis step that uses (\ref{eq:mtmMetropolis}). Third, \citet{yang2021convergence} rigorously prove superior mixing of MH over MTM albeit in the simplified setting of independence proposals.  Perhaps most importantly, \citet{tjelmeland2004using} shows that MTM's additional Metropolis step and its extra noise are simply not necessary.

Next, \citet{tjelmeland2004using} and \citet{frenkel2004speed} develop multiproposal MCMC algorithms that do not require the same binary accept-reject step as MTM. \citet{tjelmeland2004using} posits a collection of $P+1$ random variables $\TTTheta^*=(\TTheta_1^*,\dots,\TTheta^*_{P+1})$ and a probability distribution with density $q(\TTheta_p^*,\TTTheta^*)$ that describes the probability of proposing the entire set $\TTTheta^*$ given $\TTheta_p^*$.  \citet{tjelmeland2004using} augments the state space with a random variable $\frak{p} \in \{1,\dots,P+1\}$ indicating $\TTheta_\frak{p}$ as the current state of the Markov chain and iterates between generating $P$ proposals conditioned on the current state $\TTheta^{(s)}=\TTheta^*_{P+1}$ and randomly selecting a new value for $\frak{p}$ with probabilities
\begin{align}\label{eq:probsTjel}
	\left( \frac{\pi(\TTheta^*_1)q(\TTheta^*_1,\TTTheta^*)}{ \sum_{p=1}^{P+1}\pi(\TTheta^*_p)q(\TTheta^*_p,\TTTheta^*)}, \frac{\pi(\TTheta^*_2)q(\TTheta^*_2,\TTTheta^*)}{ \sum_{p=1}^{P+1}\pi(\TTheta^*_p)q(\TTheta^*_p,\TTTheta^*)}, \dots, \frac{\pi(\TTheta^*_{P+1})q(\TTheta^*_{P+1},\TTTheta^*)}{\sum_{p=1}^{P+1}\pi(\TTheta^*_p)q(\TTheta^*_p,\TTTheta^*)}\right) \, .
\end{align}
In addition to these probabilities, \citet{tjelmeland2004using} advances two further transition strategies and two structured proposal strategies, which we discuss below.  Finally, he proposes building estimators using a weighted average with weights equal to the probabilities (\ref{eq:probsTjel}) and proves that this estimator is consistent assuming the chain converges to equilibrium.  \citet{frenkel2004speed} also uses the probabilities (\ref{eq:probsTjel}) as weights to build  the same estimator but in the context of statistical physics.  \citet{frenkel2004speed} provides an informal physics argument for the correctness of his approach, stating that it fails to maintain detailed balance but does maintain `superdetailed' balance. 
More recently, \citet{calderhead2014general} adds an additional subroutine to the algorithm of \citet{tjelmeland2004using} so that, at each MCMC iteration, one simulates a random walk on the finite state space composed of the $P+1$ states using transition probabilities closely related to (\ref{eq:probsTjel}).
Ultimately, however, \citet{yang2018parallelizable} and \citet{schwedes2021rao} show that the weighted estimators of \citet{tjelmeland2004using} and \citet{frenkel2004speed} constitute Rao-Blackwellizations over the index variable $\frak{p}$ within the framework of \citet{calderhead2014general} and thus enjoy reduced variance.

Despite theoretical results favoring the use of weighted estimators that feature all proposals at each step, there are practical downsides to this approach.  Here, one must: either (1) save a massive $S\times D \times (P+1)$ tensor to memory, where $S$ is the total number of MCMC iterations; or (2) know what the estimand of interest is beforehand and maintain a running average of each MCMC iteration's contribution to the weighted estimator.
Importantly, \citet{schwedes2021rao} empirically demonstrate that the non-Rao-Blackwellized estimators of \citet{calderhead2014general} can perform poorly when a small number of jumps occur between the $P+1$ states of the proposal set.  On the one hand, this result is disconcerting because their experimental baseline of $1\times P$ jumps for $P=1000$ proposals performs much worse than $16 \times P$ jumps for the same number of proposals. Time spent iterating between proposed states is time not spent exploring the target distribution. On the other hand, these empirical results rely on a specific multiproposal scheme in which
\begin{align}\label{eq:pInd}
	\TTheta^*_1,\dots,\TTheta^*_P \stackrel{iid}{\sim} N_D(\TTheta_{P+1}^*,\SSigma) 
\end{align}
and do not necessarily extend to other multiproposal strategies.

With this in mind, we consider it worthwhile to further develop structured multiproposal methods that maintain scalability while not requiring many jumps within the proposal set. As mentioned above, \citet{tjelmeland2004using} proposes two strategies.  Proposal alternative 1, or `P1' , first generates a Gaussian random variable $\TTheta_0$ centered at the current state $\TTheta^*_{P+1}$ and next generates $P$ proposals centered at $\TTheta_0$:
\begin{align}\label{eq:p1}
	\TTheta^*_1,\dots,\TTheta^*_P \stackrel{iid}{\sim} N_D(\TTheta_{0},\SSigma) \, , \quad \TTheta_{0} \sim N_D(\TTheta_{P+1}^*,\SSigma) \, .
\end{align}
Here, the two $D$-dimensional covariance matrices take the same form so that all $\TTTheta^*=(\TTheta_1^*,\dots,\TTheta_{P+1}^*)$ share the same expected distance from the center $\TTheta_{0}$.  In some ways, the recent work of \citet{luo2019multiple} builds on this proposal strategy.  They extend this multiproposal's simple star-shaped structure, where $\TTheta_{0}$ is the center and $(\TTheta_1^*,\dots,\TTheta_{P+1}^*)$ are the flairs, to general acyclic graphs and tailor their proposals to models with certain discrete structures.  
Unlike his first proposal alternative, \citet{tjelmeland2004using}'s  Proposal alternative 2, or `P2', only applies to the scenario when $P=2$ and constructs proposal sets $\TTTheta^*=(\TTheta_1^*,\TTheta_2^*,\TTheta_3^*)$ that are equidistant from a center $\TTheta_{0}$ and `maximally spread'.  Setting $\TTheta_3^*$ to equal $\TTheta^{(s)}$, the current state of the Markov chain, this scheme satisfies the following relations:
\begin{gather}\label{eq:p2}
	\TTheta^*_1 =  \TTheta_{0} + n\v_1 \, , \quad \TTheta^*_2 =  \TTheta_{0} - \frac{n}{2} \v_1 + \frac{\sqrt{3}n}{2}  \v_2\, , \quad \TTheta_0 =  \TTheta_3^* +  \frac{n}{2} \v_1 + \frac{\sqrt{3}n}{2}  \v_2 \, ,\\ \nonumber
	\v_1 = \u_1 \, , \quad \v_2 = \frac{\left(\mathbf{I}_D - \u_1 \u_1^T\right) \u_2}{\lVert \left(\mathbf{I}_D - \u_1 \u_1^T\right) \u_2\rVert_2} \, , \quad
	n \sim N\left(0,\sigma^2/3\right) \, , \quad \u_1, \, \u_2 \stackrel{iid}{\sim} \mbox{Uniform}(\mathcal{S}^{D-1}) \, .
\end{gather}
In words, one first generates two independent random variables that are each uniformly distributed on the sphere and forces them to be orthogonal to each other.  Next, one generates the normally distributed $n$ and obtains $\TTheta_{0}$ followed by $\TTheta_1^*$ and $\TTheta_2^*$.  Straightforward calculations show that the proposals $\TTheta_1^*$, $\TTheta_2^*$ and $\TTheta_3^*$ are indeed the same distance from the center $\TTheta_{0}$. Similar calculations show that the proposals $\TTheta_1^*$, $\TTheta_2^*$ and $\TTheta_3^*$ are all equidistant from each other and together constitute an equilateral triangle with random edge lengths.

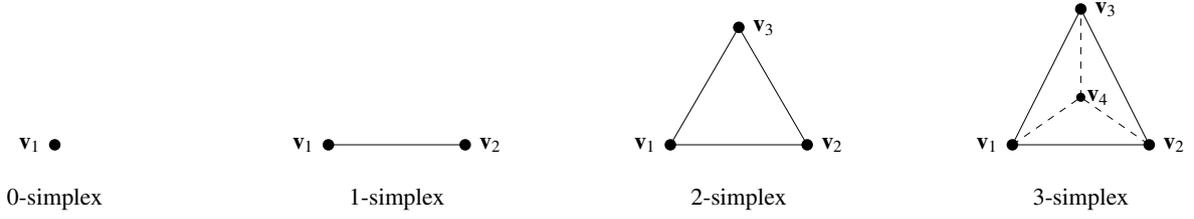
\begin{figure}[!t]
	\centering
	\scalebox{0.9}{
		\begin{tikzpicture}
			\tikzstyle{point}=[thick,draw=black,fill=black,shape=circle,inner sep=0pt,minimum width=4pt,minimum height=4pt]
			\tikzstyle{pointSmall}=[thick,draw=black,fill=black,shape=circle,inner sep=0pt,minimum width=3pt,minimum height=3pt]
			
			\node (l)[point,label={[label distance=0cm]180:$\v_1$}] at (-9,0) {};
			\node (aa)[label={0-simplex}] at (-9,-1.2) {};
			
			\node (q)[point,label={[label distance=0cm]180:$\v_1$}] at (-5,0) {};
			\node (r)[point,label={[label distance=0cm]0:$\v_2$}] at (-3,0) {};
			\node (aa)[label={1-simplex}] at (-4,-1.2) {};
			\draw (q.center) -- (r.center);

			\node (x)[point,label={[label distance=0cm]180:$\v_1$}] at (0,0) {};
			\node (y)[point,label={[label distance=0cm]0:$\v_2$}] at (2,0) {};
			\node (z)[point,label={[label distance=0cm]0:$\v_3$}] at (1,1.73) {};
			\node (aa)[label={2-simplex}] at (1,-1.2) {};
			\draw (x.center) -- (y.center) -- (z.center) -- cycle;
			
			\node (a)[point,label={[label distance=0cm]180:$\v_1$}] at (5,0) {};
			\node (b)[point,label={[label distance=0cm]0:$\v_2$}] at (7,0) {};
			\node (c)[point,label={[label distance=0cm]0:$\v_3$}] at (6,2) {};
			\node (d)[pointSmall,label={[label distance=-0.1cm]0:$\v_4$}] at (6,0.7) {};
			\draw (a.center) -- (b.center) -- (c.center) -- cycle;
			\draw[dashed] (a.center) -- (d.center) -- (b.center);
			\draw[dashed] (d.center) -- (c.center);
			\node (aa)[label={3-simplex}] at (6,-1.2) {};
		\end{tikzpicture}
	}
	\caption{Regular simplices. The 0-simplex is a point; the 1-simplex is a line connecting two points; the regular 2-simplex is an equilateral triangle; and the regular 3-simplex is a tetrahedron, the faces of which are equilateral triangles.}\label{fig:regSimps}
\end{figure}

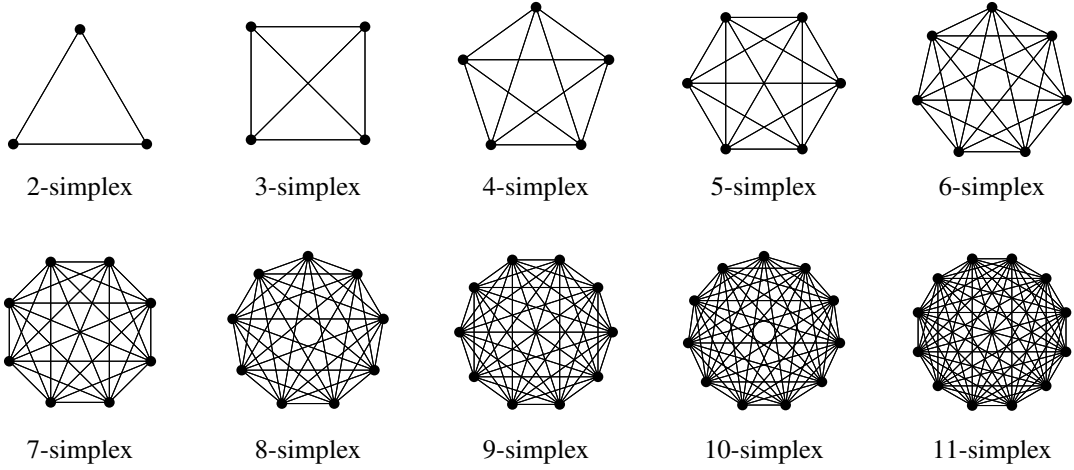
\begin{figure}[!t]
	\centering
	\scalebox{1}{
		\begin{tikzpicture}
			\tikzstyle{point}=[thick,draw=black,fill=black,shape=circle,inner sep=0pt,minimum width=4pt,minimum height=4pt]

			\node (a)[draw=none,minimum size=2cm,regular polygon,regular polygon sides=3] at (0,0) {};
			\foreach \x in {1,2,3}
			\fill (a.corner \x) circle[radius=2pt];
			\foreach \x in {1,2,3}
			\foreach \y in {1,2,3}
			\draw (a.corner \x) -- (a.corner \y);
			\node (a)[label={2-simplex}] at (0,-1.5) {};
			
			\node (a)[draw=none,minimum size=2.1cm,regular polygon,regular polygon sides=4] at (3,0.3) {};
			\foreach \x in {1,...,4}
			\fill (a.corner \x) circle[radius=2pt];
			\foreach \x in {1,...,4}
			\foreach \y in {1,...,4}
			\draw (a.corner \x) -- (a.corner \y);
			\node (a)[label={3-simplex}] at (3,-1.5) {};
			
			\node (a)[draw=none,minimum size=2cm,regular polygon,regular polygon sides=5] at (6,0.3) {};
			\foreach \x in {1,...,5}
			\fill (a.corner \x) circle[radius=2pt];
			\foreach \x in {1,...,5}
			\foreach \y in {1,...,5}
			\draw (a.corner \x) -- (a.corner \y);
			\node (a)[label={4-simplex}] at (6,-1.5) {};
			
			\node (a)[draw=none,minimum size=2cm,regular polygon,regular polygon sides=6] at (9,0.3) {};
			\foreach \x in {1,...,6}
			\fill (a.corner \x) circle[radius=2pt];
			\foreach \x in {1,...,6}
			\foreach \y in {1,...,6}
			\draw (a.corner \x) -- (a.corner \y);
			\node (a)[label={5-simplex}] at (9,-1.5) {};
			
			\node (a)[draw=none,minimum size=2cm,regular polygon,regular polygon sides=7] at (12,0.3) {};
			\foreach \x in {1,...,7}
			\fill (a.corner \x) circle[radius=2pt];
			\foreach \x in {1,...,7}
			\foreach \y in {1,...,7}
			\draw (a.corner \x) -- (a.corner \y);
			\node (a)[label={6-simplex}] at (12,-1.5) {};
			
			\node (a)[draw=none,minimum size=2cm,regular polygon,regular polygon sides=8] at (0,-3) {};
			\foreach \x in {1,...,8}
			\fill (a.corner \x) circle[radius=2pt];
			\foreach \x in {1,...,8}
			\foreach \y in {1,...,8}
			\draw (a.corner \x) -- (a.corner \y);
			\node (a)[label={7-simplex}] at (0,-5) {};
			
			\node (a)[draw=none,minimum size=2cm,regular polygon,regular polygon sides=9] at (3,-3) {};
			\foreach \x in {1,...,9}
			\fill (a.corner \x) circle[radius=2pt];
			\foreach \x in {1,...,9}
			\foreach \y in {1,...,9}
			\draw (a.corner \x) -- (a.corner \y);
			\node (a)[label={8-simplex}] at (3,-5) {};
			
			\node (a)[draw=none,minimum size=2cm,regular polygon,regular polygon sides=10] at (6,-3) {};
			\foreach \x in {1,...,10}
			\fill (a.corner \x) circle[radius=2pt];
			\foreach \x in {1,...,10}
			\foreach \y in {1,...,10}
			\draw (a.corner \x) -- (a.corner \y);
			\node (a)[label={9-simplex}] at (6,-5) {};
			
			\node (a)[draw=none,minimum size=2cm,regular polygon,regular polygon sides=11] at (9,-3) {};
			\foreach \x in {1,...,11}
			\fill (a.corner \x) circle[radius=2pt];
			\foreach \x in {1,...,11}
			\foreach \y in {1,...,11}
			\draw (a.corner \x) -- (a.corner \y);
			\node (a)[label={10-simplex}] at (9,-5) {};
			
			\node (a)[draw=none,minimum size=2cm,regular polygon,regular polygon sides=12] at (12,-3) {};
			\foreach \x in {1,...,12}
			\fill (a.corner \x) circle[radius=2pt];
			\foreach \x in {1,...,12}
			\foreach \y in {1,...,12}
			\draw (a.corner \x) -- (a.corner \y);
			\node (a)[label={11-simplex}] at (12,-5) {};

		\end{tikzpicture}
	}
	\caption{Petrie polygon projections \citep{coxeter1973regular} of higher-dimensional, regular simplices preserve symmetry but not distances.}\label{fig:petries}
\end{figure}

\citet{tjelmeland2004using} does not specify how one might extend his second proposal alternative to accommodate multiproposals when $P>2$. In the following, we accomplish this multivariate generalization by first noting that the equilateral triangle is a regular 2-simplex.  Generally, a $D$-dimensional space can hold as a many as $D+1$ points that are pairwise equidistant. If $\tilde{D}\leq D+1$, the points $\v_1, \dots, \v_{\tilde{D}}\in \mathbb{R}^D$ that satisfy
\begin{align*}
	||\v_d-\v_{d'}||_2 = \begin{cases} 
		\lambda>0\,, & d \neq d'\,,\\
		0\,, & d = d' \,,
	\end{cases}
\end{align*}
identify a regular polytope called a $\tilde{D}$-simplex. Fig.~\ref{fig:regSimps} shows the regular simplices that fit in 3-dimensional space.  Fig.~\ref{fig:petries} shows Petrie polygon projections \citep{coxeter1973regular} that are helpful for visualizing higher dimensional regular simplices.  While we use the regular simplex as our starting point, the multistep nature of (\ref{eq:p2}) makes it unclear how one might actually generate proposals when $P>2$, an endeavor that doubtless involves more random variables with more complicated relationships. To preserve simplicity in this multivariate setting, in Section \ref{sec:simp} we generate proposals using a matrix distribution over the orthogonal group $\orthog_D$, which consists of $D \times D$ matrices $\QQ$ satisfying $\QQ^T\QQ=\mathbf{I}_D$.  Because $\orthog_D$ is a compact topological group, it admits a unique normalized, left-invariant measure $\haar(\cdot)$ called the Haar measure \citep{folland2016course}. In symbols,
\begin{align*}
	\haar(\orthog_D)=1 \, , \quad  \haar \left(\QQ \mathcal{B} \right) = \haar \left( \mathcal{B}\right) 
\end{align*}
for $\QQ\in \orthog_D$ and $\mathcal{B}$ a Borel set on $\orthog_D$. The Haar measure on $\orthog_D$ is not only useful for constructing simple multivariate proposals: its uniformity helps maintain a simplified acceptance step as well (Theorem \ref{lem:symm}).

Bearing in mind the poor performance of simple multiproposal strategies when the number of jumps between proposals is small \citep{schwedes2021rao}, we are motivated to develop structured multiproposal strategies that do not require Rao-Blackwellization, or many inter-proposal jumps, and thus spend more time exploring the entire target distribution.  We therefore develop a multivariate multiproposal framework that maintains the principles of \citet{tjelmeland2004using}'s second proposal strategy while also achieving greater simplicity.  In the following, we develop the simplicial sampler and its extensions, prove algorithmic correctness and  demonstrate usefulness of our approach on a range of target distributions.  We start by deriving sufficient conditions for simplifying the acceptance probabilities (\ref{eq:probsTjel}).

\section{Simplified acceptance probabilities}

\newcommand{\dd}{\mbox{d}}

We wish to sample from the target distribution with density $\pi(\TTheta)$ for $\TTheta \in \mathbb{R}^D$.  Given the current Markov chain state $\TTheta^{(s)}$, we use the following steps to advance the Markov chain.
\begin{enumerate}
	\item Propose states $\TTTheta^*=(\TTheta^*_1, \dots, \TTheta^*_P,\TTheta^{*}_{P+1})=(\TTheta^*_1, \dots, \TTheta^*_P,\TTheta^{(s)})$  according to the distribution $q(\TTheta^{(s)},\TTTheta^*)$;
	\item Calculate densities $\pi(\TTheta^*_p)$ for $p \in \{1,\dots,P+1\}$;
	\item Randomly select a single state $\TTheta^*_p$ from the set $\TTTheta^*$ according to the simplified selection probabilities
	\begin{align}\label{eq:simpProbs}
		\left(\pi_1(\TTTheta^*), \dots, \pi_{P+1}(\TTTheta^*)\right)=	\left( \frac{\pi(\TTheta^*_1)}{ \sum_{p=1}^{P+1}\pi(\TTheta^*_p)}, \dots, \frac{\pi(\TTheta^*_{P+1})}{\sum_{p=1}^{P+1}\pi(\TTheta^*_p)}\right) \, ;
	\end{align}
	\item Set $\TTheta^{(s+1)} = \TTheta^*_{p}$.
\end{enumerate}
This transition rule uses a modified acceptance probability (\ref{eq:probsTjel}) that does not include the proposal density terms. As such, it rarely leaves the target distribution invariant. We require a further sufficient condition to guarantee that $\pi(\TTheta)$ is a stationary distribution for our Markov chain.
\begin{prop}[\citealt{tjelmeland2004using}]\label{prop:db}
	The Markov chain with the above transition rule maintains detailed balance with respect to the target distribution $\pi(\TTheta)$ if $q(\cdot,\cdot)$ satisfies
	\begin{align}\label{eq:criterion}
		q(\TTheta^*_1, \TTTheta^*) = \dots = q(\TTheta^*_p,  \TTTheta^*)= \dots = q(\TTheta^*_{P+1}, \TTTheta^*) \, ,
	\end{align}
	and so $\pi(\TTheta)$ is a stationary distribution of the chain.
\end{prop}
\begin{proof}
	Let $p(\TTheta,\dd \TTheta)$ be the transition kernel associated with the transition from an arbitrary state $\TTheta$ according to the above steps 1-4. Then
	\begin{align*}
		p(\TTheta,\dd \TTheta) = \sum_{p=1}^{P+1} \int_{(\mathbb{R}^{D})^{\otimes P}} \pi_p (\TTTheta)\, \delta_{\TTheta_p} (\dd \TTheta) \,  q(\TTheta,\TTTheta)\, \dd \TTheta_1 \dots \dd \TTheta_{P} \, .
	\end{align*}
	For  any two states $\TTheta$, $\widetilde{\TTheta}$, the following holds:
	\begin{align*}
		\pi(\TTheta)\dd\TTheta\, p(\TTheta,\dd\widetilde{\TTheta}) 
		&= \pi(\TTheta)\dd\TTheta \sum_{p=1}^{P+1} \int_{(\mathbb{R}^{D})^{\otimes P}} \pi_p (\TTTheta^*)\, \delta_{\TTheta_p^*} (\dd \widetilde{\TTheta}) \,  q(\TTheta,\TTTheta^*)\, \dd \TTheta_1^* \dots \dd \TTheta_{P}^* \\ \nonumber
		&= \pi(\TTheta)\dd\TTheta  \sum_{p=1}^{P+1} \int_{(\mathbb{R}^{D})^{\otimes P}}\frac{ \pi(\widetilde{\TTheta})}{\sum_{p'=1}^{P+1}\pi(\TTheta^*_{p'})} \, \dd \widetilde{\TTheta} \,  q(\TTheta,\TTTheta^*)\, \prod_{p'\neq p} \dd \TTheta_{p'}^* \\ \nonumber
		&= \pi(\widetilde{\TTheta})\dd\widetilde{\TTheta} \sum_{p=1}^{P+1} \int_{(\mathbb{R}^{D})^{\otimes P}}\frac{ \pi(\TTheta)}{\sum_{p'=1}^{P+1}\pi(\TTheta^*_{p'})} \, \dd \TTheta \,  q(\widetilde{\TTheta},\TTTheta^*)\, \prod_{p'\neq p} \dd \TTheta_{p'}^*  \\ \nonumber
		&= \pi(\widetilde{\TTheta})\dd\widetilde{\TTheta} \sum_{p=1}^{P+1} \int_{(\mathbb{R}^{D})^{\otimes P}} \pi_p (\TTTheta^*)\, \delta_{\TTheta_p^*} (\dd \TTheta) \,  q(\widetilde{\TTheta},\TTTheta^*)\, \dd \TTheta_1^* \dots \dd \TTheta_{P}^*  \\ \nonumber
		&= \pi(\widetilde{\TTheta})\dd \widetilde{\TTheta}\, p(\widetilde{\TTheta},\dd \TTheta) \, ,
	\end{align*}	
	where we use assumption (\ref{eq:criterion}) in the third line.
\end{proof}
Although the proof of Proposition \ref{prop:db} is new, the result itself is a special case of a result from \citet{tjelmeland2004using}, who argues that his method constitutes a Gibbs sampler that leaves the joint and marginal distributions of $\pi(\TTheta) q(\TTheta,\TTTheta^*)$ invariant.  Proposition \ref{prop:db} underlines what one may readily gather from the acceptance probabilities (\ref{eq:probsTjel}). Namely, it may sometimes prove useful to design multiproposals that maintain criterion (\ref{eq:criterion}) insofar as such proposals obviate the need for computing proposal densities.
	In the following, we show that the Gaussian centered multiproposal (\ref{eq:p1}) provides just such a multiproposal mechanism.
Under this scheme, consider the $DP$-vector $\TTheta^*_{-p}:= \mbox{vec}(\TTTheta^*_{-p})$ obtained by removing the state $\TTheta_p^*$, or the $p$th column, from the proposal set $\TTTheta^*$ and applying the vectorization operator. Integrating over all possible $\TTheta_{0}$ gives the conditional distribution of $\TTheta_{-p}^*$ given $\TTheta^*_{p}$:
\begin{align*}
	\TTheta_{-p}^* \sim N_{DP} \left(  \mathbf{1}_P \otimes \TTheta^*_{p} = 
	\begin{pmatrix}
		\TTheta_{p}^* \\
		\TTheta_{p}^* \\
		\vdots \\
		\TTheta_{p}^* 
	\end{pmatrix},\; \left(\mathbf{1}_P\mathbf{1}_P^T + \mathbf{I}_P \right) \otimes  \SSigma = 
	\begin{pmatrix}
		2 \SSigma & \SSigma & \cdots & \cdots & \SSigma \\
		\SSigma  & 2 \SSigma & \SSigma & \cdots & \SSigma \\
		\vdots   &  \SSigma & 2\SSigma & & \vdots \\
		\vdots   &  \vdots     & &\ddots & \SSigma \\
		\SSigma &   \SSigma    & \cdots & \SSigma& 2 \SSigma
	\end{pmatrix}
	\right) \, .
\end{align*}
If one does not adapt the $D \times D$ proposal covariance $\SSigma$, then one may precompute the $DP \times DP$ covariance matrix inverse in $O(D^3P^3)$ time and store it using $O(D^3P^3)$ memory.  In high-dimensional settings where such storage is not possible, one may store only the $D \times D$ inverse covariance $\SSigma^{-1}$ and fill the $DP \times DP$ inverse covariance at each step with complexity $O(D^2P^2)$ using the formula $(\mathbf{A} \otimes \mathbf{B})^{-1} = \mathbf{A}^{-1} \otimes \mathbf{B}^{-1}$.  Once one has this inverse, the inner product term in each
\begin{align*}
	q(\TTheta_{p}^*, \TTTheta^*) \propto \exp \left( - \frac{1}{2}(\TTheta^*_{-p} - \TTheta_{p}^*)^{T}  \left(\left(\mathbf{1}_P\mathbf{1}_P^T + \mathbf{I}_P \right) \otimes  \SSigma\right)^{-1} (\TTheta^*_{-p} - \TTheta_{p}^*) \right)
\end{align*}
requires $O(D^2P^2)$ floating point operations.  The upshot is an $O(D^2P^3)$ cost to compute all $P+1$ proposal densities  $q(\TTheta^*_p,\TTTheta^*)$.
But if one adapts $\SSigma$ following \citet{haario2001adaptive}, then this acceptance mechanism requires: an $O(D^3)$ matrix inversion for the updated $\SSigma$ at each step; an $O(P^2D^2)$ Kronecker product to construct the $DP \times DP$ covariance matrix; and the $O(D^2P^3)$ floating point operations to compute the $P+1$ proposal densities $q(\TTheta^*_p,\TTTheta^*)$.  Thankfully, it turns out, none of these computations are necessary.  

\begin{lemma}\label{lem:gaussCenter}
	The Gaussian centered multiproposal $q(\TTheta^{(s)},\TTTheta^*)$ defined by
	\begin{align*}
		\TTheta^*_1,\dots,\TTheta^*_P \stackrel{iid}{\sim} N_D(\TTheta_{0},\SSigma) \, , \quad \TTheta_{0} \sim N_D(\TTheta_{P+1}^*,\SSigma) \, .
	\end{align*}
	satisfies symmetry relation \eqref{eq:criterion}, namely
	\begin{align*}
		q(\TTheta^*_1, \TTTheta^*) = \dots = q(\TTheta^*_p,  \TTTheta^*)= \dots = q(\TTheta^*_{P+1}, \TTTheta^*) \, .
	\end{align*}
\end{lemma}
\begin{proof}
	Let $p(\cdot,\cdot)$ denote the generic probability density function evaluated at the right and conditioned on the left, and let $\phi(\cdot,\cdot)$ an analogous multivariate Gaussian density with dependence on the covariance suppressed. For two arbitrary indices $p, p' \in \{1,\dots,P+1\}$,
	\begin{align*}
		q(\TTheta^*_p,\TTTheta^*) - q(\TTheta^*_{p'},\TTTheta^*)&= \int_{\mathbb{R}^D} p(\TTheta_p^*,\TTheta_{0}) \,  p(\TTheta_{0},\TTheta_{-p}^*)\, \dd \TTheta_{0} - \int_{\mathbb{R}^D} p(\TTheta_p^*,\TTheta_{0}) \,  p(\TTheta_{0},\TTheta_{-p}^*)\, \dd \TTheta_{0}\\ \nonumber 
		&= \int_{\mathbb{R}^D}\left( p(\TTheta_p^*,\TTheta_{0}) \,  p(\TTheta_{0},\TTheta_{-p}^*)  -  p(\TTheta_{p'}^*,\TTheta_{0}) \,  p(\TTheta_{0},\TTheta_{-p'}^*)\right) \dd \TTheta_{0} \\ \nonumber 
		&=\int_{\mathbb{R}^D} \left( \phi(\TTheta_p^*,\TTheta_{0}) \prod_{p''\neq p}\phi(\TTheta_{0},\TTheta_{p''}^*) -\phi(\TTheta_{p'}^*,\TTheta_{0}) \prod_{p'''\neq p'}\phi(\TTheta_{0},\TTheta_{p'''}^*)\right) \dd \TTheta_{0} \\ \nonumber
		&= \int_{\mathbb{R}^D} \left(  \prod_{p''=1}^{P+1} \phi(\TTheta_{0},\TTheta_{p''}^*) - \prod_{p'''=1}^{P+1}\phi(\TTheta_{0},\TTheta_{p'''}^*)\right) \dd \TTheta_{0} = 0 \, .
	\end{align*}
	In the third line, we use the Gaussianity and conditional independence structure of the multiproposal. In the fourth line, we use the symmetry of the Gaussian density function.
\end{proof}

This simple result fills in a gap in \citet{tjelmeland2004using}: namely, it relates the symmetry in a proposal strategy to a simplified acceptance mechanism \eqref{eq:simpProbs}.  Moreover, subsequent work \citep{glatt} builds on this result by examining cases in which structured proposals lead to simplified acceptance probabilities.  In the following, we develop a multivariate extension of the multiproposal \eqref{eq:p2} and show that it also satisfies symmetry relation \eqref{eq:criterion} and thus also qualifies for the simplified acceptance mechanism \eqref{eq:simpProbs}.

\section{Simplicial sampling}\label{sec:simp}

We develop a proposal mechanism and a few variations that satisfy symmetry relation \eqref{eq:criterion} and help our basic algorithm preserve reversibility with respect to the target $\pi(\TTheta)$. This method uses certain mathematical objects introduced in Section \ref{sec:intro}: namely, the regular $D$-simplex, the orthogonal group $\orthog_D$ and the uniform Haar measure over the orthogonal group $\haar(\orthog_D)$.  Start with a regular simplex identified by a set of fixed vertices $\v_1, \dots, \v_{D+1}\in \mathbb{R}^D$ that satisfies $\v_{D+1}=\boldsymbol{0}$.  Equating the number of new proposals $P$ with the state space dimension $D$, the simplicial sampler generates and chooses from proposals $\TTheta^*_1,\dots,\TTheta^*_{D+1}$ as follows.
\begin{enumerate}
	\item Sample $\QQ \sim \haar(\orthog_D)$.
	\item Rotate and translate the simplicial vertices $(\v_1,\dots,\v_D,\boldsymbol{0}) \longmapsto \QQ (\v_1,\dots,\v_D,\boldsymbol{0}) + \TTheta^{(s)} =: (\TTheta_1^*,\dots,\TTheta_{D+1}^*)$.
	\item Draw a single sample $\TTheta^*_d$ from $(\TTheta^*_1,\dots,\TTheta_{D+1}^*)$ with probability proportional to $\pi(\TTheta^*_d)$.
	\item Set $\TTheta^{(s+1)} = \TTheta^*_{d}$.
\end{enumerate}
In short, the simplicial sampler advances its Markov chain by randomly rotating a set of $D$ equidistant vertices about the current state $\TTheta^{(s)}$ (Fig.~\ref{fig:rotate1}) and randomly selecting the next state from among all $D+1$ vertices of the resulting simplex.

A modified simplicial sampler that only generates a single proposal reduces to Barker's acceptance criterion \citep{barker1965monte}
\begin{align*}
	\pi(\TTheta^*)/\left(\pi(\TTheta^*)+\pi(\TTheta^{(s)})  \right)
\end{align*}
with a symmetric proposal distribution. Just as Barker's criterion is suboptimal with respect to the Peskun ordering \citep{peskun1973optimum}, our criterion is also suboptimal. With this in mind, \citet{tjelmeland2004using} provides an algorithm that manipulates the acceptance probabilities \eqref{eq:probsTjel}, of which our simplified acceptance probabilities are a special case. 

\begin{figure}[!t]
	\centering
	\scalebox{0.8}{
		\begin{tikzpicture}
			\node (a)[label={[label distance=0cm]180:$\TTheta_1^*$}] at (0.5,0.5) {};
			\node (b)[label={[label distance=0cm]0:$\TTheta_3^*$}] at (4.3,3) {};
			\node (c)[label={[label distance=0cm]100:$\TTheta_2^*$}] at (1.6,4) {};
			\node (d)[label={[label distance=0cm]0:{$\TTheta^{(s)}$}}] at (4.2,-1) {};
			\draw (a.center) -- (b.center) -- (c.center) -- cycle;
			\draw (a.center) -- (d.center) -- (b.center);
			\draw [dashed](d.center) -- (c.center);
			
			\draw [-latex, thick, rotate=30] (6.5,-1.5) arc [start angle=-160, end angle=160, x radius=1.5cm, y radius=0.5cm];
			\draw[thick,->] (6.3,0) arc (145:35:1.6);
			\node (e)[label={[label distance=0cm]0:{$\times$}}] at (7.1,1.4) {};
			
			\node (a)[label={[label distance=0cm]0:$\TTheta_1^*$}] at (15.5,-0.5) {};
			\node (b)[label={[label distance=0cm]180:$\TTheta_3^*$}] at (13,3) {};
			\node (c)[label={[label distance=0cm]0:$\TTheta_2^*$}] at (14.8,1.8) {};
			\node (d)[label={[label distance=0cm]180:{$\TTheta^{(s)}$}}] at (11,-1) {};
			\draw (a.center) -- (c.center) -- (d.center) -- cycle;
			\draw (b.center) -- (c.center) -- (d.center) -- cycle;
			\draw (a.center) -- (d.center) -- (b.center);
			\draw [dashed](a.center) -- (b.center);
		\end{tikzpicture}
	}
	\caption{A simplicial sampling multiproposal for $D=3$. A proposal set is obtained by rotating three simplex vertices about current state $\TTheta^{(s)}$.}\label{fig:rotate1}
\end{figure}
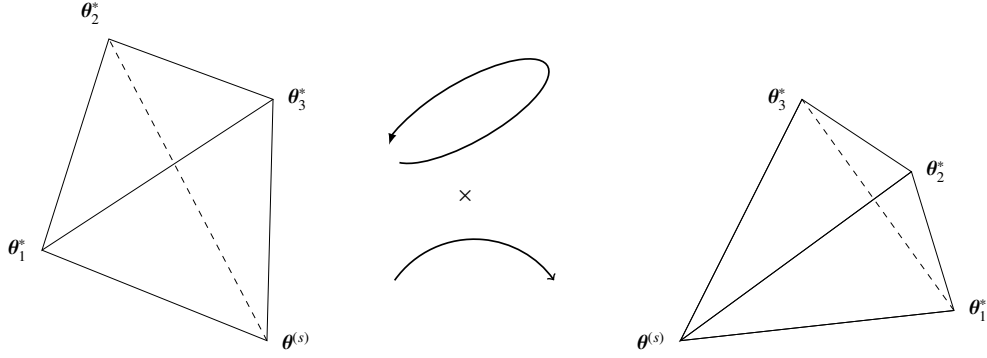

Importantly, the random rotation step satisfies symmetry relation \eqref{eq:criterion}.

\begin{lemma}\label{lem:symm}
	The simplicial sampler multiproposal $q(\TTheta^{(s)},\TTTheta^*)$ defined by
	\begin{align*}
		(\v_1,\dots,\v_D,\boldsymbol{0}) \longmapsto \QQ (\v_1,\dots,\v_D,\boldsymbol{0}) + \TTheta^{(s)} =: (\TTheta_1^*,\dots,\TTheta_{D+1}^*) \quad \mbox{for} \quad \QQ \sim \haar\left( \orthog_D \right)
	\end{align*}
	satisfies symmetry relation \eqref{eq:criterion}, namely
	\begin{align*}
		q(\TTheta^*_1, \TTTheta^*) = \dots = q(\TTheta^*_d,  \TTTheta^*)= \dots = q(\TTheta^*_{D+1}, \TTTheta^*) \, .
	\end{align*}
\end{lemma}
\begin{proof}
	Start by translating coordinates so that $\TTheta^{(s)} = \boldsymbol{0}$.
	We want to show that $q(\boldsymbol{0},\TTTheta^*)=q(\TTheta^*_d,\TTTheta^*)$ for an arbitrary $\TTheta_d^*\in \TTTheta^*$. We focus on the case $d\leq D$, since the alternative case is trivial.
	Since the Haar measure $\haar\left( \orthog_D\right)$ is uniform on $\orthog_D$, we only need to show the existence of a rotation $\QQ_d\in \orthog_D$ that maps
	\begin{align*}
		(\v_1,\dots,\v_D) \longmapsto (\TTheta_1^*,\dots,\TTheta_{d-1}^*,\boldsymbol{0} , \TTheta_{d+1}^*,\dots,\TTheta_{D}^*) \, .
	\end{align*} 
	The key is recognizing that the unique reflection $\mathbf{R}$ that exchanges $\boldsymbol{0}$ and $\TTheta_{d}^*$ leaves all other vertices unchanged because they are equidistant from $\boldsymbol{0}$ and $\TTheta_{d}^*$ and  inhabit the hyperplane that splits them evenly. In symbols, $\mathbf{R}\TTheta^*_{d'}= \TTheta^*_{d'}$ for $d'\neq d$. Thus
	\begin{align*}
		\QQ^{T} \mathbf{R} \TTheta^*_{d'} = \QQ^{T}  \TTheta^*_{d'} = \v_{d'} \, , \quad\forall d' \neq d \, ,
	\end{align*}
	and 
	\begin{align*}
		\QQ^{T} \mathbf{R} \boldsymbol{0}=  \QQ^{T} \TTheta^*_d = \v_d \, .
	\end{align*}
	It follows that $( \QQ^{T} \mathbf{R})^T=\mathbf{R}\QQ$ is the desired rotation matrix $\QQ_d$.
	
\end{proof}

Fig.~\ref{fig:correctness} provides a stylized visualization of this argument in the $2D$ setting when $\QQ$ represents a 60$^{\circ}$ rotation clockwise about $\TTheta^{(s)}$.
Taken together, Lemma \ref{lem:symm} and Proposition \ref{prop:db} assert that the simplicial sampler maintains detailed balance with respect to the target distribution.
\begin{thm}\label{thm:first}
	If the simplicial sampler's multiproposal is $\pi$-irreducible and absolutely continuous with respect to $\pi(\TTheta)$, then $\pi(\TTheta)$ is the unique stationary and limiting distribution of the Markov chain, and the strong law of large numbers holds.
\end{thm}
\begin{proof}
	The result is a simple combination of the sampler's reversibility with respect to $\pi(\TTheta)$ and the aperiodicity that arises from its self-loop, i.e., the ability to set $\TTheta^{(s+1)}$ to $\TTheta^{(s)}$ \citep{tierney1994markov}.
\end{proof}

While $\pi$-irreducibility is a major hypothesis of Theorem \ref{thm:first}, we do not formally prove that the simplicial sampler exhibits generally irreducible behavior.  That said, there are two intuitive reasons for thinking that the sampler produces generally irreducible Markov chains.  First, the simplicial sampler can return to an arbitrarily small ball around a given state within just two steps with positive probability: the first iteration accepts a state $\lambda$ away from the current state, and the second iteration accepts a state in the neighborhood of the preceding state.  Second, the simplicial sampler can amble arbitrarily far away from a given state in the manner that a distracted geometry student can walk a compass clear off their desk.
The need for a formally justified sampler leads us to randomly scale the proposal simplex in Section \ref{sec:gauss}.

\usetikzlibrary{arrows.meta}
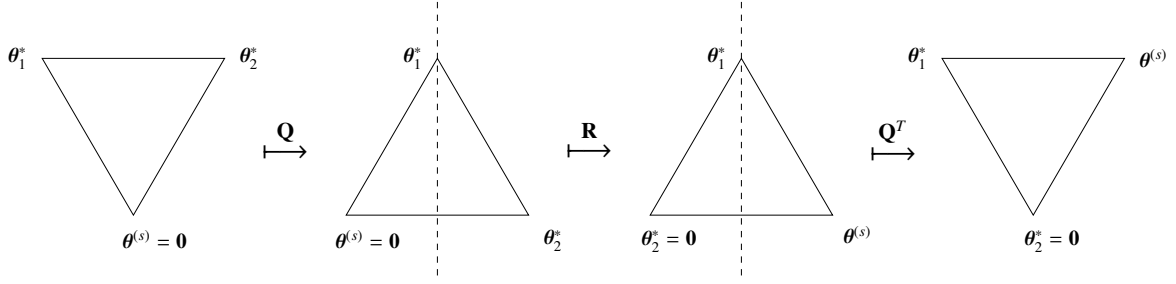
\begin{figure}[!t]
	\centering
	\scalebox{0.8}{
		\begin{tikzpicture}
			
			\node (a)[label={[label distance=0cm]270:\hspace{2em}$\TTheta^{(s)}=\boldsymbol{0}$}] at (-3,0) {};
			\node (b)[label={[label distance=0cm]180:$\TTheta_1^*$}] at (-4.5,2.598076) {};
			\node (c)[label={[label distance=0cm]0:$\TTheta_2^*$}] at (-1.5,2.598076) {};
			\draw (a.center) -- (b.center) -- (c.center) -- cycle;
			
			\node (e)[scale=1.5] at (-0.5,1.25) {$\stackrel{\QQ}{\longmapsto}$};
			
			\node (a)[label={[label distance=0cm]270:\hspace{1em}$\quad\TTheta^{(s)}=\boldsymbol{0}$}] at (0.5,0) {};
			\node (b)[label={[label distance=0cm]180:$\TTheta_1^*$}] at (2,2.598076) {};
			\node (c)[label={[label distance=0cm]315:$\TTheta_2^*$}] at (3.5,0) {};
			\draw (a.center) -- (b.center) -- (c.center) -- cycle;
			
			\draw [dashed] (2,-1) -- (2,1+2.598076);

  	\node (e)[scale=1.5] at (4.5,1.25) {$\stackrel{\mathbf{R}}{\longmapsto}$};
%
%

			\node (a)[label={[label distance=0cm]270:\hspace{1.75em}$\TTheta_2^*=\boldsymbol{0}$}] at (5.5,0) {};
			\node (b)[label={[label distance=0cm]180:$\TTheta_1^*$}] at (7,2.598076) {};
			\node (c)[label={[label distance=-0.1cm]315:$\TTheta^{(s)}$}] at (8.5,0) {};
			\draw (a.center) -- (b.center) -- (c.center) -- cycle;
			
			\draw [dashed] (7,-1) -- (7,1+2.598076);

			\node (e)[scale=1.5] at (9.5,1.25) {$\stackrel{\QQ^T}{\longmapsto}$};
			
			\node (a)[label={[label distance=0cm]270:\hspace{1.75em}$\TTheta_2^*=\boldsymbol{0}$}] at (11.8,0) {};
			\node (b)[label={[label distance=0cm]0:$\TTheta^{(s)}$}] at (13.3,2.598076) {};
			\node (c)[label={[label distance=0cm]180:$\TTheta_1^*$}] at (10.3,2.598076) {};
			\draw (a.center) -- (b.center) -- (c.center) -- cycle;		
			
			%
			%
		\end{tikzpicture}
	}
	\caption{Illustrating key argument of Lemma \ref{lem:symm} for $D=2$. $\QQ$ is the simplicial sampler's random rotation about $\TTheta^{(s)}$ and $\mathbf{R}$ is a reflection that exchanges a given proposal $\TTheta_2^{*}$ with the current state $\TTheta^{(s)}$ while preserving all other proposals.}\label{fig:correctness}
\end{figure}

\subsection{Variations}

\subsubsection{Gaussian simplicial sampling}\label{sec:gauss}

\newcommand{\one}{\boldsymbol{1}}

If we scale the simplex edge lengths by the square root of a chi-square random variable, the proposals will have multivariate Gaussian marginal distributions.
The orthogonal group $\orthog_D$ acts transitively on the sphere $\mathcal{S}^{D-1}$, i.e., it takes any point on $\mathcal{S}^{D-1}$ to any other \citep{Lee2012}.  For this reason, left multiplying a single point on $\mathcal{S}^{D-1}$ by a Haar distributed rotation matrix $\QQ \sim \haar(\orthog_D)$ is one method for drawing a sample from the sphere's uniform distribution. On the other hand, it is a well known fact that one may obtain a standard $D$-dimensional Gaussian vector by uniformly sampling a single point on the sphere $\mathcal{S}^{D-1}$ and scaling that point by the square root of an independently sampled $\chi^2$ random variable with $D$ degrees of freedom. Randomly scaling the simplicial edge length in a similar manner provides the following variation that we call the Gaussian simplicial sampler.
\begin{enumerate}
	\item Independently sample $r\sim \chi^2(D)$ and  $\QQ\sim\haar(\orthog_D)$.
	\item Rotate, scale and translate the vertices $(\v_1,\dots,\v_D,\boldsymbol{0}) \longmapsto \sqrt{r} \QQ (\v_1,\dots,\v_D,\boldsymbol{0}) + \TTheta^{(s)} =: (\TTheta_1^*,\dots,\TTheta_{D+1}^*)$.
	\item Draw a single sample $\TTheta^*_d$ from $(\TTheta^*_1,\dots,\TTheta_{D+1}^*)$ with probability proportional to $\pi(\TTheta^*_d)$.
	\item Set $\TTheta^{(s+1)} = \TTheta^*_{d}$.
\end{enumerate}
By the above discussion, each of the proposals $\TTheta_d^*$ for $d\in\{1,\dots,D\}$ has a Gaussian marginal distribution with mean $\TTheta^{(s)}$ and covariance matrix $\lambda^2\mathbf{I}_D$, where $\lambda$ is still the simplicial edge length.  Because of its random scaling, a stronger result than Theorem \ref{thm:first} applies to the Gaussian simplicial sampler.
\begin{cor}\label{cor:one}
	Any distribution $\pi(\TTheta)$ with support $\mathbb{R}^D$ is the unique stationary and limiting distribution of the Gaussian simplicial sampler that targets it, and the strong law of large numbers holds. 
\end{cor}
\begin{proof}
	All that needs to be established is $\pi$-irreducibility and absolute continuity. 	The latter condition is accomplished by hypothesis.  The former condition is accomplished by the sufficient condition \citep{tierney1994markov} that
	\begin{align*}
		\int_A q(\TTheta_0,\TTTheta^*) \dd \TTheta^{*}_d > 0 \, ,  \quad \mbox{if}  \quad  \int_A \pi(\TTheta) \dd \TTheta > 0, \quad  d \in\{1,\dots,D\}\,  , \; A \subset \mathbb{R}^D\, .
	\end{align*}
	This condition is itself accomplished by the unbounded nature of the multiproposal's $\chi^2$ scaling.  
	 All edges are scaled by the same random variable, and the symmetry of the proposal simplex is therefore maintained. Because the random scaling is independent of the random rotation, the conditions within the proof of Lemma \ref{lem:symm} remain exactly the same. 
\end{proof}
In Section \ref{sec:results}, we show that its random scaling gives the Gaussian simplicial sampler improved performance for a multimodal target distribution.

\subsubsection{Preconditioned simplicial sampling}\label{sec:pc}
\newcommand{\CC}{\mathbf{C}}

For $\CC \succ \boldsymbol{0}$, a $D\times D$ and real-valued matrix, the \emph{preconditioned simplicial sampler} takes the following steps.
\begin{enumerate}
	\item Sample $\QQ\sim\haar(\orthog_D)$.
	\item Rotate, precondition and translate the vertices $(\v_1,\dots,\v_D,\boldsymbol{0}) \longmapsto \CC^{1/2} \QQ (\v_1,\dots,\v_D,\boldsymbol{0}) + \TTheta^{(s)} =: (\TTheta_1^*,\dots,\TTheta_{D+1}^*)$.
	\item Draw a single sample $\TTheta^*_d$ from $(\TTheta^*_1,\dots,\TTheta_{D+1}^*)$ with probability proportional to $\pi(\TTheta^*_d)$.
	\item Set $\TTheta^{(s+1)} = \TTheta^*_{d}$.
\end{enumerate}
The matrix $\CC^{1/2}\QQ$ is itself a rotation in the Euclidean space equipped with inner product
$\langle \v, \v' \rangle_{\CC^{-1}} := \v^T\CC^{-1} \v'$ (Fig.~\ref{fig:rotate2}).  
\begin{cor}\label{cor:two}
	The guarantees of Theorem \ref{thm:first} apply to the preconditioned simplicial sampler as well.
\end{cor}
\begin{proof}
	The preconditioned simplicial sampler maintains detailed balance with respect to its target distribution $\pi(\TTheta)$ because it satisfies symmetry relation \eqref{eq:criterion}.
	The proof takes the exact same form as that of Lemma \ref{lem:symm}, replacing $\QQ_d$ with $\CC^{1/2}\QQ_d$.
\end{proof}
Corollaries \ref{cor:one} and \ref{cor:two} combine to guarantee a unique stationary distribution and strong law of large numbers for a preconditioned Gaussian simplicial sampler, and we explore this algorithm in Section \ref{sec:results}.

\begin{figure}[!t]
	\centering
	\resizebox{\linewidth}{!}{
		\begin{tikzpicture}
			\node (a)[label={[label distance=-0.1cm]180:$\TTheta_1^*$}] at (0,0.1) {};
			\node (b)[label={[label distance=-0.1cm]0:$\TTheta_3^*$}] at (3.9,-0.1) {};
			\node (c)[label={[label distance=-0.1cm]80:$\TTheta_2^*$}] at (2.1,1.2) {};
			\node (d)[label={[label distance=-0.3cm]140:{\footnotesize$\TTheta^{(s)}$}}] at (2,0.4) {};
			\draw (a.center) -- (b.center) -- (c.center) -- cycle;
			\draw[dashed] (a.center) -- (d.center) -- (b.center);
			\draw[dashed] (d.center) -- (c.center);
			\node (f)[label={[label distance=-0.1cm]80:$90^\circ$}] at (5,-1.3) {};
			
			\node (a)[label={[label distance=-0.1cm]290:$\TTheta_1^*$}] at (9.3,-0.2) {};
			\node (b)[label={[label distance=-0.2cm]70:$\TTheta_3^*$}] at (9.4,1.2) {};
			\node (c)[label={[label distance=-0.1cm]180:$\TTheta_2^*$}] at (6,0.6) {};
			\node (d) at (7.5,0.4) {};
			\draw (a.center) -- (b.center) -- (c.center) -- cycle;
			\draw[dashed] (a.center) -- (d.center) -- (b.center);
			\draw[dashed] (d.center) -- (c.center);
			\node (f)[label={[label distance=-0.1cm]80:$90^\circ$}] at (11,-1.3) {};
			
			\node (a)[label={[label distance=-0.2cm]45:$\TTheta_1^*$}] at (15.2,1) {};
			\node (b)[label={[label distance=-0.2cm]135:$\TTheta_3^*$}] at (11.4,1.2) {};
			\node (c)[label={[label distance=-0.1cm]270:$\TTheta_2^*$}] at (13.3,-0.1) {};
			\node (d) at (13.4,0.4) {};
			\draw (a.center) -- (b.center) -- (c.center) -- cycle;
			\draw[dashed] (a.center) -- (d.center) -- (b.center);
			\draw[dashed] (d.center) -- (c.center);

			\draw[thick,->] (4.5,-1) arc (215:325:1);
			\draw[thick,->] (10.5,-1) arc (215:325:1);
		\end{tikzpicture}
	}
	\caption{The preconditioned simplicial sampler's rotations about $\TTheta^{(s)}$ with respect to a non-standard Euclidean metric.  To a viewer relying on the standard Euclidean metric, fixed-length edges appear to lengthen and shorten as the simplex is rotated 90$^\circ$ according to the standard metric about the axis perpendicular to the page.}\label{fig:rotate2}
\end{figure}
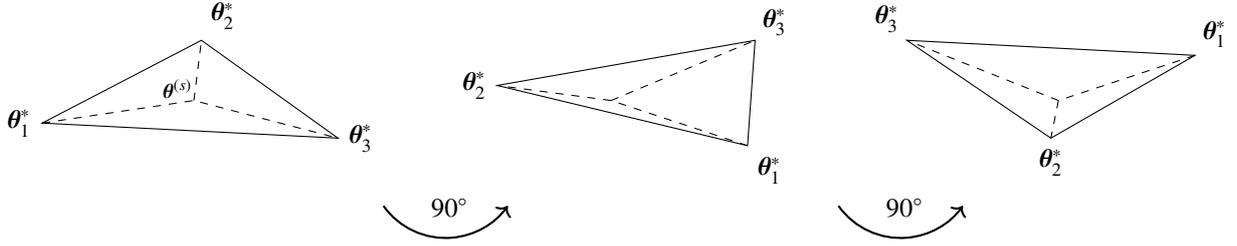

\subsubsection{Adaptive simplicial sampling}\label{sec:adapt}

As has been established for other MCMC algorithms, it is often useful to adapt the simplicial sampler multiproposal. We employ two distinct methods of adaptation in Section \ref{sec:results}.  Both methods use the diminishing adaptation strategy \citep{rosenthal2011optimal}, which sacrifices guarantees of a strong law of large numbers but maintains a weak law.  The first adapts all simplex edge lengths while keeping them equal so as to achieve a desired `acceptance' rate. Although the simplicial sampler does not use a classical accept-reject step, we can check whether state $\TTheta^{(s+1)}$ equals $\TTheta^{(s)}$ and call this an acceptance or rejection accordingly.  The second adapts the preconditioning matrix $\CC$ following the proposal covariance adaptation outlined in \citet{haario2001adaptive}.  In practice, we find it useful to combine these two adaptation schemes when implementing the preconditioned simplicial sampler.

\subsubsection{Extra-dimensional simplicial sampling}

The final variation for the simplicial sampler uses a number of proposals that is larger than the dimensionality of the target distribution, i.e., $P>D$. Let $\mathbf{W}=\left(\mathbf{I}_{D\times D} ,\: \boldsymbol{0}_{D\times (P-D)}\right)$ be the wide $D \times P$ matrix for which multiplication selects the first $D$ elements of any $P$-vector. Letting each $\v_p$, $p\in\{1,\dots,P\}$, be a $P$-dimensional vector, the extra-dimensional simplicial sampler works as follows.
\begin{enumerate}
	\item Sample $\QQ\sim\haar(\orthog_P)$.
	\item Rotate, precondition and translate the vertices $(\v_1,\dots,\v_P,\boldsymbol{0}_P) \longmapsto  \QQ (\v_1,\dots,\v_P,\boldsymbol{0}_P) + (\TTheta^{(s)T},\boldsymbol{0}_{P-D}^T)^T =: (\TTheta_1^*,\dots,\TTheta_{P+1}^*)$.
	\item Draw a single sample $\mathbf{W}\TTheta^*_p$ from $\mathbf{W}(\TTheta^*_1,\dots,\TTheta_{P+1}^*)$ with probability proportional to $\pi(\mathbf{W}\TTheta^*_p)$.
	\item Set $\TTheta^{(s+1)} = \mathbf{W}\TTheta^*_{p}$.
\end{enumerate}
One may interpret the application of $\mathbf{W}$ as a form of preconditioning with preconditioner $\CC^{1/2}=\mathbf{W}$. In Section \ref{sec:pc}, we argue that the guarantees of Lemma \ref{lem:symm} apply to the preconditioned simplicial sampler because $\CC^{1/2}\QQ$ is a random rotation in Euclidean space with inner product $\langle \v, \v' \rangle_{\CC^{-1}} := \v^T\CC^{-1} \v'$.  Following this analogy, we have
\begin{align*}
	\CC=\mathbf{W}^T\mathbf{W}= 
	\begin{pmatrix}
		\mathbf{I}_D & \boldsymbol{0}_{D\times(P-D)} \\
		\boldsymbol{0}_{(P-D)\times D} & 	\boldsymbol{0}_{(P-D)\times (P-D)}
	\end{pmatrix} = \CC^{+} \, ,
\end{align*}
for $\CC^{+}$ the Moore-Penrose inverse \citep{holbrook2018differentiating}.  Due to the degeneracy of $\CC^{+}$, the resulting form  $\langle \v, \v' \rangle_{\CC^{+}} := \v^T\CC^{+} \v'$ is no longer an inner product, and we cannot assert that the extra-dimensional algorithm inherits the promises of  Lemma \ref{lem:symm}.
Nonetheless, limited empirical experiments do show the algorithm providing accurate results, and we find it helpful to visualize the sampler in Section \ref{sec:ed}.

\section{Computational considerations}

The real-world performance of an algorithm depends on its computational complexity, the language it is written in and the hardware it is implemented on.  It also depends on how well one codes it up, e.g., efficient use of random access memory and parallelization when available. We separate these considerations into multiproposal and acceptance mechanisms, and the exact weight one assigns to the two steps will depend on the computational complexity of the target distribution.

\subsection{Multiproposal mechanisms}

\subsubsection{The simplicial sampler multiproposal}\label{sec:ssMult}
The simplicial sampler requires the $O(D^3)$ generation of a random rotation $\QQ \sim \haar(\orthog_D)$.  The most common way to perform this step is to generate a $D\times D$ matrix of standard normals and then perform the QR factorization. Once one has the orthogonal `Q' matrix, one right multiplies it by a diagonal matrix consisting of the signs of the diagonal of the `R' matrix, guaranteeing that the obtained matrix $\QQ$ is the unique orthogonal matrix corresponding to the $\mathbf{R}$ with positive diagonal elements.  Thus, one obtains $\QQ \sim \haar(\orthog_D)$ using the $O(D^3)$ QR factorization.  There are additional algorithms for sampling from the Haar distribution over $\orthog_D$ that provide speedups while still maintaining cubic complexity \citep{stewart1980efficient,anderson1987generation}. In particular, \citet{stewart1980efficient} achieves a 50\% speedup over the QR factorization based algorithm by iteratively pre-multiplying random Householder matrices of increasing dimensions. We are interested in the fact that, when one only requires the random rotation and not the matrix itself, \citet{stewart1980efficient}'s algorithm only requires $O(D^2)$ floating point operations to randomly rotate a $D$-vector. On the one hand, the simplicial sampler multiproposal still requires $O(D^3)$ floating point operations because it must apply this random rotation to the $D$ $D$-vectors that identify the regular simplex.  On the other hand, one could apply this algorithm in an embarrassingly parallel manner to the $D$ non-zero $D$-vectors that contribute to the regular simplex. Furthermore, the entirety of the algorithm relies on matrix-matrix multiplications and should be extremely fast and scalable on a modern GPU: \citet{li2013gpu} use a (now outdated) Tesla C1060 GPU to score 10,000-fold speedups over a single-core $IJK$-algorithm implementation written in \textsc{C++}.  Such a fast GPU implementation of \citet{stewart1980efficient}'s algorithm requires significant engineering that is a contribution in its own right, so we focus on conventional implementations in the following.

Finally, even for serial implementations, the language of implementation can eclipse theoretical computational complexity. We find that an \textsc{R} based implementation of \citet{stewart1980efficient}'s algorithm is much slower than the QR factorization based method implemented in the \textsc{PRACMA} \citep{pracma} \textsc{R} package.  The QR decomposition in \textsc{R} calls the \textsc{Fortran} library \textsc{LAPACK} to perform the numerical linear algebra for the QR decomposition, and its precompiled and memory efficient implementation outperform interpreted \textsc{R}.

\subsubsection{Other multiproposals}

If one precomputes $\SSigma^{1/2}$, the multiproposals \eqref{eq:pInd} and \eqref{eq:p1} are of complexity $O(D^2P)$, with the majority of time spent multiplying the $D \times D$ square-root covariance $\SSigma^{1/2}$ by a $D \times P$ matrix of standard normals. If one adopts the preconditioning strategy of \citet{haario2001adaptive}, the complexity changes to $O(D^3+D^2P)$ accounting for a Cholesky decomposition. If one specifies $\SSigma = \sigma^2 \mathbf{I}_D$, then the complexity drops to $O(DP)$, and the new bottleneck becomes generation of i.i.d.~standard normals.

Just as for the simplicial sampler's random rotations, one may parallelize the matrix multiplications for these multiproposals.  One may also parallelize the standard normal generations using the parallel pseudo random number generators of \citet{salmon2011parallel} that \textsc{TensorFlow} uses \citet{abadi2016tensorflow} to great success.  

\subsection{Acceptance mechanisms}

The simplified acceptance mechanism \eqref{eq:simpProbs} requires $D$ target density evaluations.  These density evaluations are embarrassingly parallel, but there are some target distributions for which this kind of parallelization is not necessary.  When one jointly infers a multivariate Gaussian random variable and its covariance matrix using Metropolis-within-Gibbs, one must invert the updated covariance every time one updates the Gaussian variable. Once one has inverted the covariance matrix, evaluating the multivariate Gaussian density over $D$ proposals is almost as fast as evaluating the density over a single proposal.  That said, this intuition also applies to Metropolis-within-Gibbs with a simple random walk Metropolis update: having updated and inverted the covariance matrix, one may perform many Metropolis updates of the Gaussian random variable at relatively little additional cost.
	In general, however, careful parallelization should help accelerate the acceptance step for multiproposal methods in a way unavailable to single-proposal methods. For example, \citet{holbrook2021quantum} uses a quantum computer to achieve quadratic speedups that reduce the number of target evaluations from $\mathcal{O}(P)$ to $\mathcal{O}(\sqrt{P})$.

The simplified acceptance mechanism also avoids potentially costly computations of the $P+1$ proposal densities $q(\TTheta^*_p,\TTTheta^*)$ required for the general acceptance mechanism \eqref{eq:probsTjel}.  For example, the i.i.d.~multiproposal \eqref{eq:pInd} requires an additional $O(P^2)$ computational cost to obtain the terms
\begin{align*}
	q(\TTheta^*_p,\TTTheta^*)  = \prod_{p'\neq p} \exp \left( -\frac{1}{2}(\TTheta^*_p-\TTheta^*_{p'})^T \SSigma^{-1} (\TTheta^*_p-\TTheta^*_{p'})\right), \quad  p\in\{1,\dots, P+1\} \, .
\end{align*}
On the other hand, Lemma \ref{lem:gaussCenter} shows that the centered Gaussian multiproposal \eqref{eq:p1} leads to a simplified acceptance mechanism that results in prodigious theoretical speedups thoroughly discussed in Section \ref{sec:simp}.

\section{Results}\label{sec:results}

In the following, we use (PC-)Simpl, (PC-)RWM and (PC-)MTM to indicate the (preconditioned) simplicial sampler, random walk and multiple-try Metropolis algorithms, respectively.  Although PC-RWM is better known as the adaptive Metropolis algorithm of \citet{haario2001adaptive}, we use this notation to clarify comparisons.  ESS and ESSs denote effective sample size and effective sample size per second, respectively.  All experiments use the \textsc{R} language \citep{ihaka1996r} and the \textsc{R} package \textsc{pracma} \citep{pracma} to produce random rotations, which itself uses the QR decomposition of a matrix with Gaussian entries (Section \ref{sec:ssMult}). All MTM simulations use $D$ proposals for comparison. Finally, all results herein rely on a conventional single-core implementation.

\subsection{Empirically optimal scaling}

\begin{figure}[t!]
	\centering
	\includegraphics[width=\linewidth]{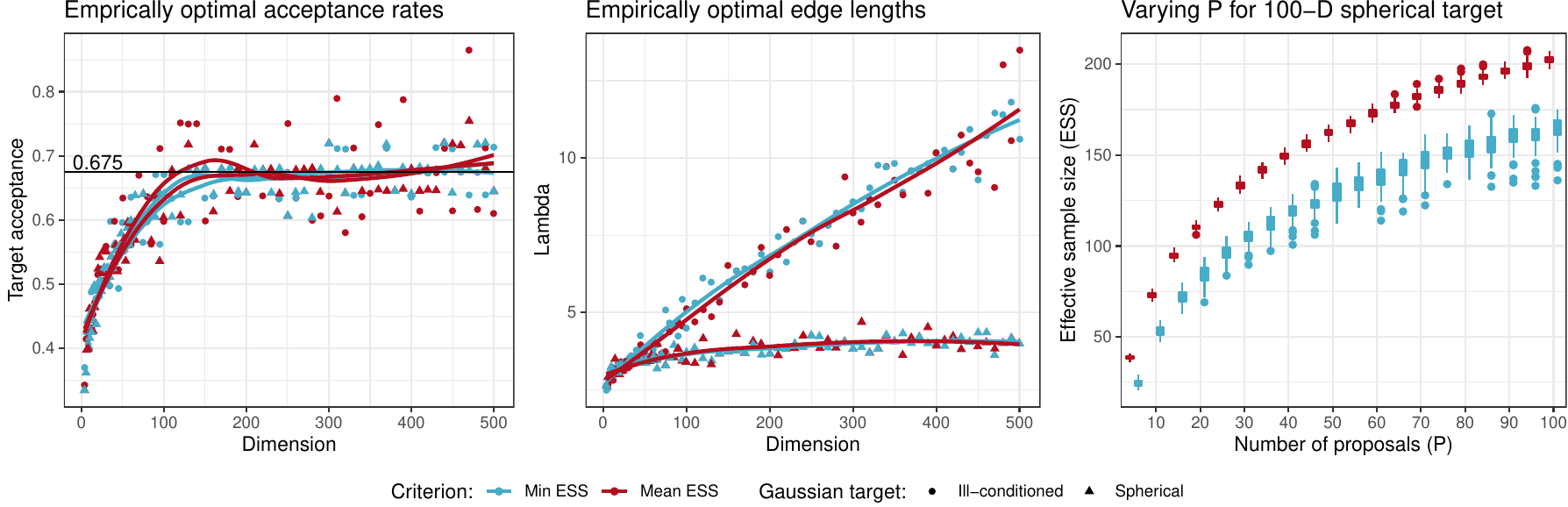}
	\caption{The left two figures present acceptance rates and simplex edge lengths that empirically maximize mean and minimum effective sample sizes (ESS) across target dimensions. For spherical and ill-conditioned targets, we use vanilla and preconditioned simplicial sampler, respectively. The right figure presents ESS results for a sampler modified to use a lower dimensional simplex.}\label{fig:optimalScaling}
\end{figure}

\begin{figure}[t!]
	\centering
	\includegraphics[width=\linewidth]{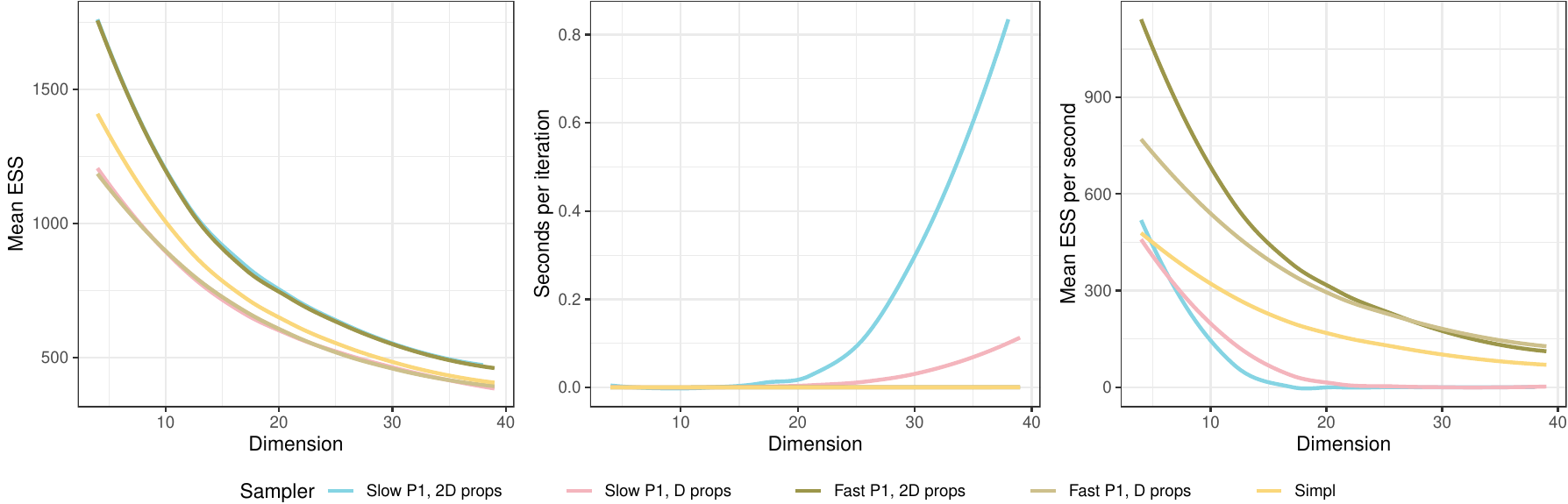}
	\caption{Performance of the simplicial sampler (Simpl) and the `P1' proposal mechanism \eqref{eq:p1}, with both $D$ and $2D$ proposals and both fast and slow implementations (Section \ref{sec:simp}).  All samplers target a standard multivariate Gaussian distribution, run for 10,000 iterations and adapt to a 50\% target acceptance rate.}\label{fig:p1p2}
\end{figure}

\subsubsection{Acceptance rates and edge lengths}
We first investigate Simpl and PC-Simple performance for Gaussian targets of varying dimensions (4 through 500) as a function of acceptance rate and simplex edge length.  For each dimension, we employ 20 independent MCMC runs of 10,000 iterations targeting one of 20 acceptance rates uniformly spaced between 0.2 and 0.95 (Section \ref{sec:adapt}).  The `ill-conditioned' Gaussian target has diagonal covariance with elements 1 through $D$.  
The left two panels of Fig.~\ref{fig:optimalScaling} plot the acceptance rates or edge lengths that maximize minimum and mean ESS out of the 20 targeted.  For both Simpl and PC-Simpl, optimal acceptance rates increase with dimensionality until flattening out at around 0.675. For Simpl, optimal edge lengths quickly flatten out at around 3. For PC-Simpl, optimal edge lengths grow with target dimension.

\subsubsection{Number of proposals}
In general, Simpl uses a simplex with $D+1$ nodes, the maximal number available in $D$ dimensions.  The right panel of Fig.~\ref{fig:optimalScaling} shows ESS results when lower dimensional simplices (with fewer nodes) are used to generate fewer proposals ($P<D$).  In practice, one may implement such a variation by replacing the initial set of pre-rotated vectors $(\v_1,\dots,\v_D,\boldsymbol{0})$ with $(\v_1,\dots,\v_P,\boldsymbol{0})$, again letting $P<D$. Algorithmic details otherwise remain the same.
We use 100 independent MCMC runs of 10,000 iterations each to generate minimum and mean ESS across target dimensions.  Both measures improve with greater numbers of proposals.  We do not compare ESSs because the additional cost of a simplex node arises from a mere $\order{D^2}$ matrix-vector multiplication, compared to the $\order{D^3}$ generation of a random rotation.

\subsection{Performance comparisons}

\begin{figure}[t!]
	\centering
	\includegraphics[width=\linewidth]{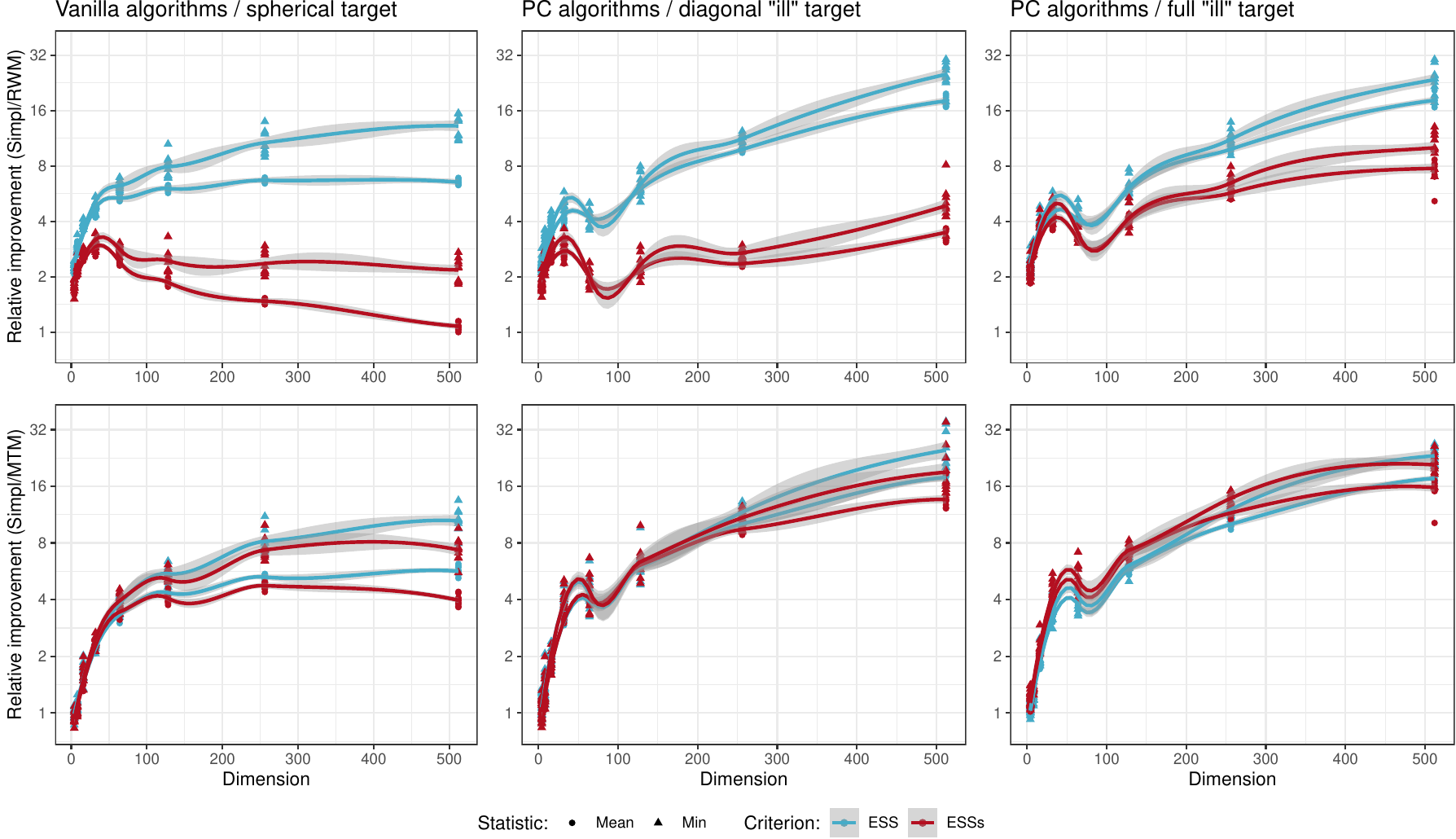}
	\caption{Relative improvement of simplicial sampler (Simpl) against random walk Metropolis (RWM) and multipiple-try Metropolis (MTM) as measured by effective sample size (ESS) and ESS per second (ESSs). We apply preconditioning (PC) in the exact same manner to all algorithms. }\label{fig:gaussianComp}
\end{figure}

\subsubsection{Accelerated centered Gaussian multiproposal}

Fig.~\ref{fig:p1p2} compares the simplicial sampler performance to that of the centered Gaussian sampler \eqref{eq:p1} using $P=D$ and $P=2D$ proposals. The `slow' Gaussian centered multiproposal calculates all $P+1$ proposal probabilities $q(\TTheta_p^*,\TTTheta^*)$, whereas the `fast' sampler uses the simplified acceptance probabilities \eqref{eq:simpProbs} following Lemma \ref{lem:gaussCenter}.  Samplers target a standard spherical normal distribution and a 50\% acceptance rate for 10,000 iterations.  In terms of ESS, doubling the number of proposals appears to benefit the centered Gaussian sampler, especially in lower dimensions.  The fast implementation of the centered Gaussian sampler already starts to outperform the slow implementations in lower dimensions, and this relative performance gain only increase with dimensions.  Across all dimensions, the ESSs trajectory of the simplicial sampler remains firmly in between those of the fast and slow centered Gaussian samplers.

\subsubsection{Gaussian targets}

We compare vanilla algorithms on spherical targets of dimensions $D=$4, 8, 16, 32, 64, 128, 256 and 512, all with identity for covariance matrix, and preconditioned algorithms on targets with ill-conditioned diagonal and full covariance matrices.  Ill-conditioned matrices all have spectra ranging uniformly from 1 to $D$. RWM and MTM both use optimal scaling of \citep{gelman1997weak}. For the latter, this results in an asymptotic 30\% acceptance rate as $D$ gets large \citep{bedard2012scaling}.  Simpl adapts to empirically optimal acceptance rates of Fig.~\ref{fig:optimalScaling}.  Fig.~\ref{fig:gaussianComp} shows relative ESS and ESSs results from 10 independent MCMC runs of 10,000 iterations each per algorithm.  Since ESS is a univariate measure, we report means and minima with respect to target dimensions.

\subsubsection{Bimodal target}

\begin{figure}[t!]
	\centering
	\includegraphics[width=\linewidth]{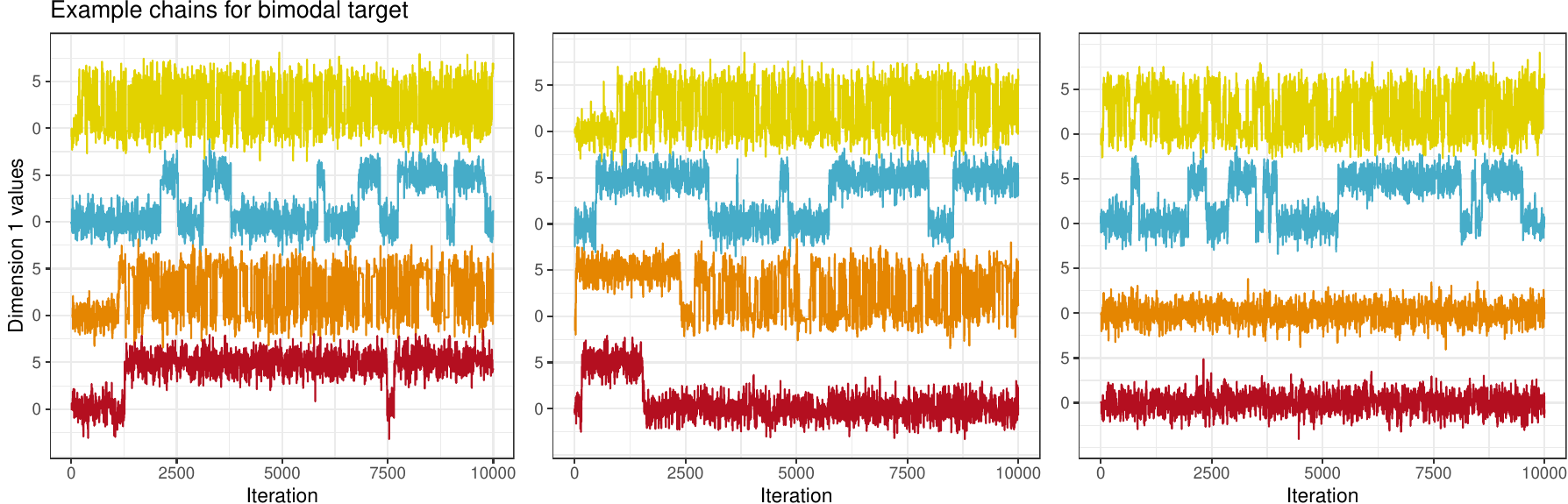}
	
	\vspace{0.5em}
	
	\includegraphics[width=0.7\linewidth]{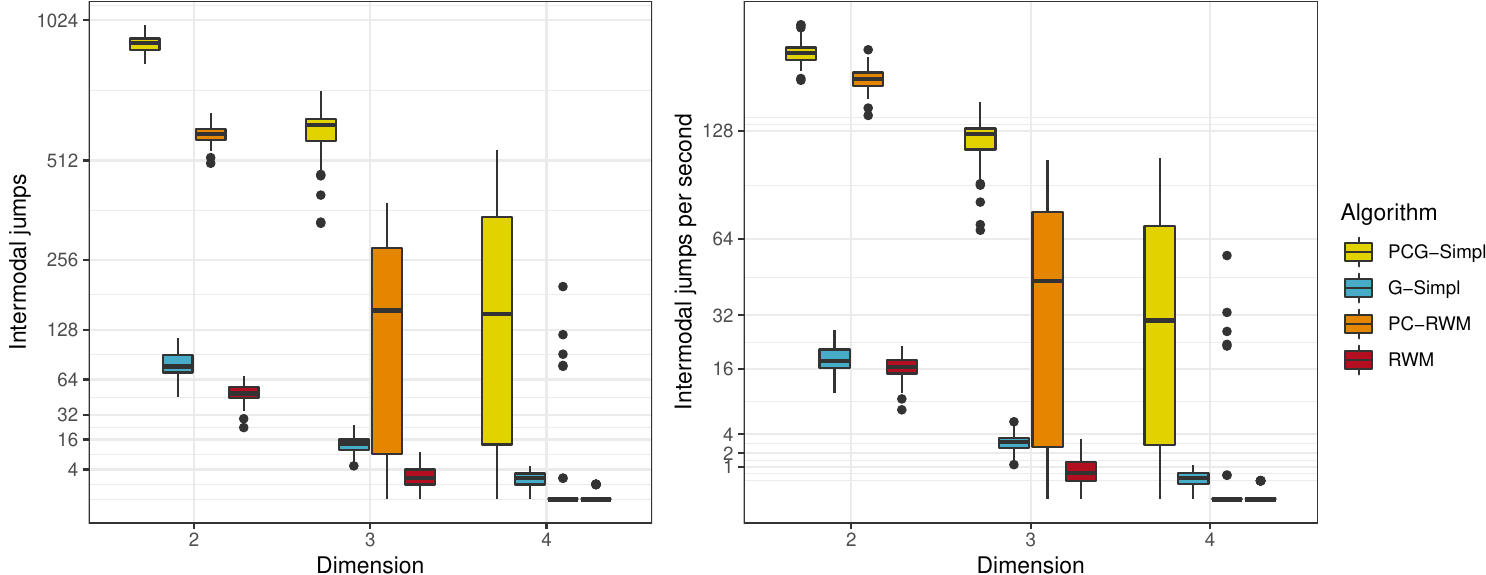}
	\caption{The top figures show stacked traceplots for the first dimension of a 3D Gaussian mixture target. The bottom figures show the number of intermodal jumps for 100,000 iterations presented over 100 independent MCMC runs. PCG- and G-Simpl are preconditioned and non-preconditioned Gaussian simplicial samplers compared to the same for random walk Metropolis (RWM). }\label{fig:bimodal}
\end{figure}

In general, Gaussian simplicial sampling (G-Simpl) (Section \ref{sec:gauss}) and its preconditioned form (PCG-Simpl) perform slightly worse ($10\%$ less ESS) than Simpl and PC-Simpl on unimodal targets, but G-Simpl and PCG-Simpl perform better on multimodal targets. We compare G-Simpl and PCG-Simpl to RWM and PC-RWM on mixture of Gaussian targets of dimensions 2, 3 and 4.  Mixture probabilities are $0.5$-$0.5$ and covariances are identity, but modal means are $\boldsymbol{0}$ and $\boldsymbol{5}= (5,\dots,5)$. Thus, modes become farther apart with increasing dimension.  Fig.~\ref{fig:bimodal} shows stacked traceplots and boxplots for the number of intermodal jumps recorded from 100 independent MCMC runs of 100,000 iterations each.

\subsubsection{Gaussian process classification}

\begin{table}[!t]
	\centering
	\caption{Effective sample sizes (ESS) and ESS per second (ESSs) from 100 independent MCMC runs of 100,000 iterations each for Gaussian process classification of 48 state election results from 2016 U.S. presidential election. Standard errors of empirical mean estimates accompany in parentheses. Simpl, RWM and MTM are vanilla simplicial sampler, random walk and multiple-try Metropolis algorithms.  PC implies preconditioning.  Its and time to $err=10$ represent iterations and seconds to achieve a misclassification total of 10, a rough measure of convergence speed. Winners in blue.}\label{tab:gp}
	\vspace{-1em}
	\resizebox{\textwidth}{!}{\begin{tabular}{@{}llllllll@{}}
			\toprule
			Algorithm & Mean ESS $\TTheta$  & Min ESS $\TTheta$  & Its to $err=10$  & ESS $\eta^2$  & ESS $\xi^2$  & ESS $\rho^2$  & ESS $\sigma^2$  \\ 
			\midrule
			Simpl & 1571.50  (15.6) & 256.27 (6.4)     & \cellcolor{trevorblue!15}73.05  (1.8) & 671.39  (10.2) & 2463.76  (76.6) & 3703.05  (50.0) & 452.30  (6.7) \\ 
			RWM & 436.34  (2.3) & 68.36  (2.3)       & 276.60  (8.3) & 258.63  (5.7) & 1407.88  (61.8) & 1837.47  (31.3) & 134.91  (2.3) \\ 
			MTM & 1195.22    (12.8) &     252.13       (7.3) &  1928.75   (20.5) & 676.75 (12.8)  &2332.81    (84.9) & 3490.12 (58.41) &366.86   (6.3) \\ 
			PC-Simpl &\cellcolor{trevorblue!15}1865.18  (14.4) & \cellcolor{trevorblue!15}1200.50  (20.4) & 76.39  (2.2) & \cellcolor{trevorblue!15}1312.98  (18.0) & \cellcolor{trevorblue!15}4583.21  (95.5) & \cellcolor{trevorblue!15}4222.10  (50.7) & \cellcolor{trevorblue!15}532.49  (6.5) \\ 
			PC-RWM & 362.99  (4.5) & 205.16  (5.3) & 284.27  (10.5) & 418.99  (10.9) & 1381.00  (37.6) & 1654.35  (29.8) & 125.12  (3.0) \\ 
			PC-MTM & 1256.07    (15.9)&      755.42    (16.4) &  1873.89   (21.1) & 951.97 (15.3) & 3840.86  (88.3) & 3736.27 (48.0) & 272.08   (6.3) \\
			\midrule
			& Mean ESSs $\TTheta$  & Min ESSs $\TTheta$  & Time to $err=10$  & ESSs $\eta^2$  & ESSs $\xi^2$  & ESSs $\rho^2$  & ESSs $\sigma^2$  \\ 
			\midrule
			Simpl & \cellcolor{trevorblue!15}5.62  (0.06) & 0.92  (0.02) & \cellcolor{trevorblue!15}15.87 (0.98) & 2.40  (0.04) & 8.81  (0.28) & \cellcolor{trevorblue!15}13.24  (0.18) & \cellcolor{trevorblue!15}1.62  (0.02) \\ 
			RWM & 2.50  (0.02) & 0.39  (0.01) & 145.41  (9.69) & 1.48  (0.03) & 8.09  (0.36) & 10.55  (0.19) & 0.77  (0.01) \\ 
			MTM & 3.85       (0.04) &    0.81      (0.02)& 11674.88  (252.45)& 2.18 (0.04) & 7.52   (0.27) & 11.25 (0.19) & 1.18 (0.02) \\ 
			PC-Simpl & 5.09  (0.13) & \cellcolor{trevorblue!15}3.30  (0.10) & 23.99  (1.46) & \cellcolor{trevorblue!15}3.58  (0.10) & \cellcolor{trevorblue!15}12.59  (0.43) & 11.53  (0.32) & 1.45  (0.04) \\ 
			PC-RWM & 1.44  (0.04) & 0.81  (0.03) & 268.77  (39.96) & 1.65  (0.06) & 5.48  (0.22) & 6.54  (0.21) & 0.49  (0.02) \\ 
			PC-MTM & 3.87       (0.05) &    2.33      (0.05)& 11540.43   (269.22)& 2.93 (0.05) & 11.85  (0.28) & 11.53 (0.15) & 0.84 (0.02) \\ 
			\bottomrule
	\end{tabular}}
\end{table}

48 states use a winner-take-all rule when assigning electoral college votes during the U.S. presidential election.  We use a Gaussian process (GP) with logit link \citep{williams1998bayesian} to regress binary outcomes (Trump or Clinton) over state center (latitude and longitude) and population data for the 2016 presidential election.  For latent variables $\TTheta=(\theta_1,\dots,\theta_{48})$, the prior covariance of $\theta_i$ and $\theta_j$ is
\begin{align*}
	\kappa (\x_i,\x_j) = \xi^2 + \eta^2 \exp\left(-\rho^2 ||\x_i-\x_j||_2^2 \right) + \sigma^2 \delta_{ij} \, ,
\end{align*}
where $\x_i$ and $\x_j$ are their respective predictor vectors. We infer latent variables $\TTheta$  using (PC-)Simple, (PC-)RWM and (PC-)MTM. We let Simpl algorithms target an acceptance rate of 0.5 following the results of Fig.~\ref{fig:optimalScaling}, while we let RWM, MTM and PC-MTM target 0.234, 0.3 and 0.4. The former value is optimal for RWM \citep{rosenthal2011optimal}.The last two values appear advantageous in preliminary tests.   We infer other parameters using univariate slice samplers \citep{neal2003slice}. Table \ref{tab:gp} shows results from 100 independent runs of 100,000 iterations each.  We initialize all chains to misclassify all states and use a total misclassification of 10 as rough convergence threshold.  While Simple and MTM require multiple target evaluations at each step, they only must compute the GP inverse covariance once for all such evaluations.  These results only establish effectiveness of the simplicial sampler relative to a few basic samplers and do not establish superiority relative to all samplers. Gradient-based MCMC methods could likely outperform the simiplicial sampler, but we abstain from such comparisons in order to focus on `apples-to-apples' comparisons among non-gradient methods.

\subsection{Visualizing the extra-dimensional simplicial sampler}\label{sec:ed}

Because of the diminishing returns (Fig.~\ref{fig:optimalScaling}, right) and $O(D^3)$ computational cost associated with increasing the number of simplex vertices, the extra-dimensional simplicial sampler is not practical in the absence of advanced parallelization techniques. Moreover, the rotated vertices no longer maintain equal distances from each other in the low-dimensional space over which the target is defined.  Fig.~\ref{fig:hyper} shows three successive iterations of the extra-dimensional simplicial sampler with 1000 proposals targeting a 2D Gaussian distribution. The sampler successfully moves to the high-density region of the space in these three steps.  We observe that rotating and projection the simplex vertices results in proposal point clouds that look plausibly Gaussian distributed, apparently lacking any geometric structure.

\begin{figure}[t!]
	\centering
	\includegraphics[width=\linewidth]{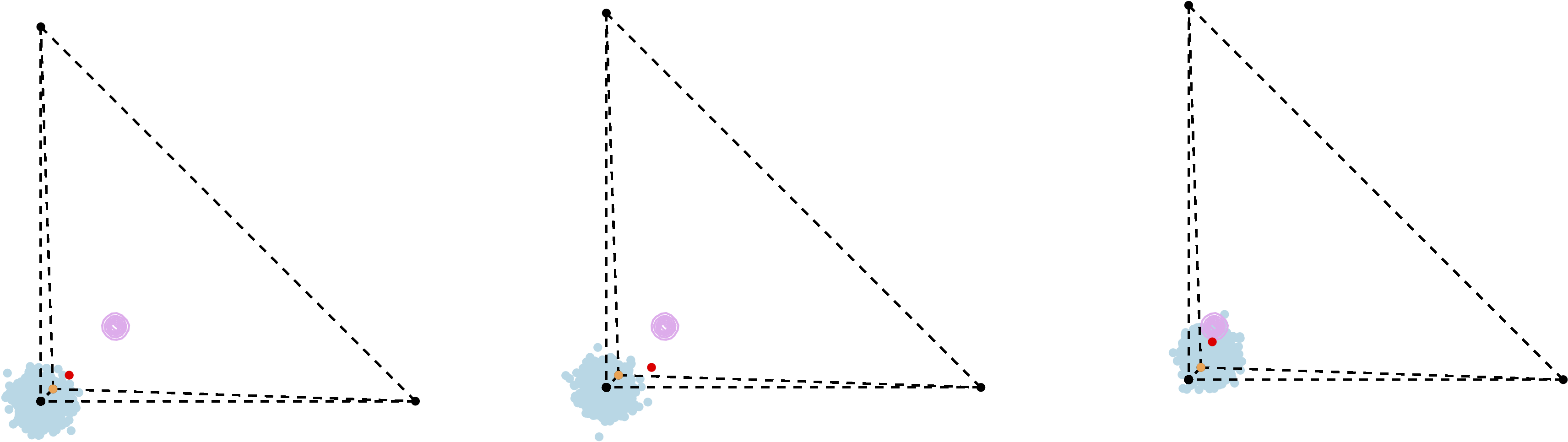}
	\caption{Successive iterations of an extra-dimensional simplicial sampler using 1000 proposals to sample a bivariate Gaussian target. For each step, the yellow dot is the initial position, the black dots are the projected un-rotated simplex vertices, the blue dots are the proposals (or projected rotated vertices), and the red dots are the accepted proposals. The dashed lines represent simplex edges, and the purple contours represent the target distribution. Finally, 998 unrotated vertices project down to the same bottom-left black point on the plane.}\label{fig:hyper}
\end{figure}

\section{Discussion}\label{sec:disc}

We have developed the simplicial sampler, a multivariate extension of the multiproposal mechanism \eqref{eq:p2} from \citet{tjelmeland2004using}, and proven that this extension maintains detailed balance using simplified acceptance probabilities. We have shown that a variation on this algorithm, the Gaussian simplicial sampler, maintains a unique stationary distribution and an appropriate law of large numbers.  Furthermore, we have proven that the centered Gaussian proposal mechanism \eqref{eq:p1} from \citet{tjelmeland2004using} enables the same simplified acceptance probability that leads to significant speedups.  We have thoroughly discussed the computational issues that one must consider when implementing these algorithms. Among others, these considerations include: theoretical computational complexity, hardware, programming language, efficient use of memory and parallelization.  Finally, we have thoroughly compared performance of sequential implementations on a collection of target distributions. That said, we note that these algorithms are well suited for parallelization enabled by, e.g., a GPU or a quantum computer \citep{holbrook2021quantum}.

We are interested in two lines of further investigation.  First, we have shown that two structurally distinct multiproposal mechanisms have symmetries that enable simplified acceptance probabilities.  Are these the only such mechanisms or are there others?  Can one combine these in ways that preserve the simplified acceptance probabilities?  Finally, do analogues for these methods exist in discrete settings, e.g., for lattice models?  Second, we are interested in the quality of convergence for the broader class of multiproposal MCMC algorithms.  Based on simulations, it would not be surprising for geometric ergodicity to hold for these algorithms. Moreover, one might expect for mixing speeds to improve as the number of proposals increases.

\section*{Acknowledgments}
\noindent
This work was supported by NIH grant K25 AI153816 and NSF grant DMS 2152774. We thank Nathan Glatt-Holtz, Cecilia Mondaini and Justin Krometis for helpful discussions and advice.

\appendix

\section{Empirical validation}\label{sec:empVal}

\begin{figure}[!t]
	\centering
	\includegraphics[width=\textwidth]{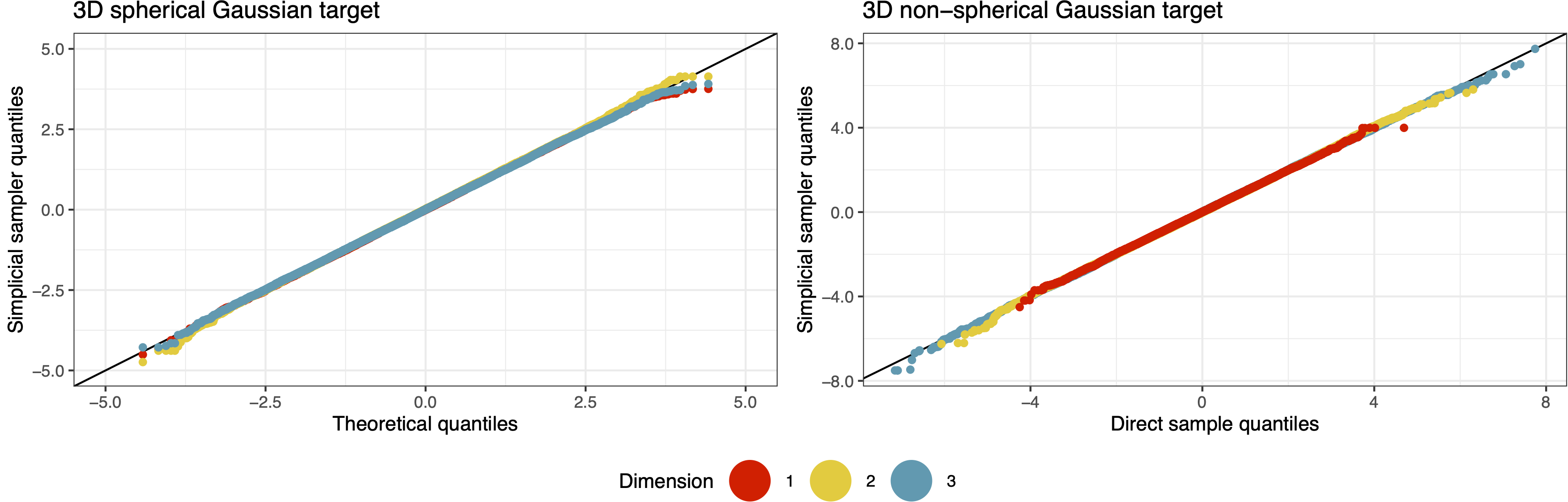}
	\caption{Quantile-quantile plots from 100,000 vanilla simplicial sampler iterations after removing 10,000 iterations as burn-in.}\label{fig:empVal}
\end{figure}

Fig.~\ref{fig:empVal} shows quantile-quantile (QQ) plots for the vanilla simplicial sampler for two different 3D Gaussian target distributions, one of which is spherical and the other of which has diagonal covariance with standard deviations 1, 2 and 3.  The QQ plots demonstrate that the empirical quantiles of the simplicial sampler adhere closely to the quantiles of respective target distributions.

\bibliographystyle{myjmva}
\bibliography{refs}

\begin{thebibliography}{32}
\expandafter\ifx\csname natexlab\endcsname\relax\def\natexlab#1{#1}\fi
\providecommand{\bibinfo}[2]{#2}
\ifx\xfnm\relax \def\xfnm[#1]{\unskip,\space#1}\fi
\bibitem[{Abadi et~al.(2016)Abadi, Agarwal, Barham, Brevdo, Chen, Citro,
  Corrado, Davis, Dean, Devin et~al.}]{abadi2016tensorflow}
\bibinfo{author}{M.~Abadi}, \bibinfo{author}{A.~Agarwal},
  \bibinfo{author}{P.~Barham}, \bibinfo{author}{E.~Brevdo},
  \bibinfo{author}{Z.~Chen}, \bibinfo{author}{C.~Citro}, \bibinfo{author}{G.~S.
  Corrado}, \bibinfo{author}{A.~Davis}, \bibinfo{author}{J.~Dean},
  \bibinfo{author}{M.~Devin}, et~al., \bibinfo{title}{Tensorflow: Large-scale
  machine learning on heterogeneous distributed systems},
  \bibinfo{journal}{arXiv preprint arXiv:1603.04467}  (\bibinfo{year}{2016}).
\bibitem[{Anderson et~al.(1987)Anderson, Olkin and
  Underhill}]{anderson1987generation}
\bibinfo{author}{T.~W. Anderson}, \bibinfo{author}{I.~Olkin},
  \bibinfo{author}{L.~G. Underhill}, \bibinfo{title}{Generation of random
  orthogonal matrices}, \bibinfo{journal}{SIAM Journal on Scientific and
  Statistical Computing} \bibinfo{volume}{8} (\bibinfo{year}{1987})
  \bibinfo{pages}{625--629}.
\bibitem[{Barker(1965)}]{barker1965monte}
\bibinfo{author}{A.~A. Barker}, \bibinfo{title}{Monte carlo calculations of the
  radial distribution functions for a proton? electron plasma},
  \bibinfo{journal}{Australian Journal of Physics} \bibinfo{volume}{18}
  (\bibinfo{year}{1965}) \bibinfo{pages}{119--134}.
\bibitem[{B{\'e}dard et~al.(2012)B{\'e}dard, Douc and
  Moulines}]{bedard2012scaling}
\bibinfo{author}{M.~B{\'e}dard}, \bibinfo{author}{R.~Douc},
  \bibinfo{author}{E.~Moulines}, \bibinfo{title}{Scaling analysis of
  multiple-try mcmc methods}, \bibinfo{journal}{Stochastic Processes and their
  Applications} \bibinfo{volume}{122} (\bibinfo{year}{2012})
  \bibinfo{pages}{758--786}.
\bibitem[{Borchers(2021)}]{pracma}
\bibinfo{author}{H.~W. Borchers}, \bibinfo{title}{pracma: Practical Numerical
  Math Functions}, \bibinfo{year}{2021}. \bibinfo{note}{R package version
  2.3.3}.
\bibitem[{Calderhead(2014)}]{calderhead2014general}
\bibinfo{author}{B.~Calderhead}, \bibinfo{title}{A general construction for
  parallelizing metropolis- hastings algorithms}, \bibinfo{journal}{Proceedings
  of the National Academy of Sciences} \bibinfo{volume}{111}
  (\bibinfo{year}{2014}) \bibinfo{pages}{17408--17413}.
\bibitem[{Coxeter(1973)}]{coxeter1973regular}
\bibinfo{author}{H.~S.~M. Coxeter}, \bibinfo{title}{Regular {P}olytopes},
  \bibinfo{publisher}{Courier Corporation}, \bibinfo{address}{New York},
  \bibinfo{year}{1973}.
\bibitem[{Folland(2016)}]{folland2016course}
\bibinfo{author}{G.~B. Folland}, \bibinfo{title}{A {C}ourse in {A}bstract
  {H}armonic {A}nalysis}, volume~\bibinfo{volume}{29}, \bibinfo{publisher}{CRC
  press}, \bibinfo{address}{Boca Raton, Florida}, \bibinfo{year}{2016}.
\bibitem[{Frenkel(2004)}]{frenkel2004speed}
\bibinfo{author}{D.~Frenkel}, \bibinfo{title}{Speed-up of monte carlo
  simulations by sampling of rejected states}, \bibinfo{journal}{Proceedings of
  the National Academy of Sciences} \bibinfo{volume}{101}
  (\bibinfo{year}{2004}) \bibinfo{pages}{17571--17575}.
\bibitem[{Gelman et~al.(1997)Gelman, Gilks and Roberts}]{gelman1997weak}
\bibinfo{author}{A.~Gelman}, \bibinfo{author}{W.~R. Gilks},
  \bibinfo{author}{G.~O. Roberts}, \bibinfo{title}{Weak convergence and optimal
  scaling of random walk {M}etropolis algorithms}, \bibinfo{journal}{The
  {A}nnals of {A}pplied {P}robability} \bibinfo{volume}{7}
  (\bibinfo{year}{1997}) \bibinfo{pages}{110--120}.
\bibitem[{Glatt-Holtz et~al.(2022)Glatt-Holtz, Holbrook, Krometis and
  Mondaini}]{glatt}
\bibinfo{author}{N.~E. Glatt-Holtz}, \bibinfo{author}{A.~J. Holbrook},
  \bibinfo{author}{J.~A. Krometis}, \bibinfo{author}{C.~F. Mondaini},
  \bibinfo{title}{Parallel mcmc algorithms: Theoretical foundations, algorithm
  design, case studies}, \bibinfo{year}{2022}.
\bibitem[{Haario et~al.(2001)Haario, Saksman and Tamminen}]{haario2001adaptive}
\bibinfo{author}{H.~Haario}, \bibinfo{author}{E.~Saksman},
  \bibinfo{author}{J.~Tamminen}, \bibinfo{title}{An adaptive {M}etropolis
  algorithm}, \bibinfo{journal}{Bernoulli} \bibinfo{volume}{7}
  (\bibinfo{year}{2001}) \bibinfo{pages}{223--242}.
\bibitem[{Hastings(1970)}]{hastings1970monte}
\bibinfo{author}{W.~K. Hastings}, \bibinfo{title}{{Monte Carlo sampling methods
  using Markov chains and their applications}}, \bibinfo{journal}{Biometrika}
  \bibinfo{volume}{57} (\bibinfo{year}{1970}) \bibinfo{pages}{97--109}.
\bibitem[{Holbrook(2018)}]{holbrook2018differentiating}
\bibinfo{author}{A.~Holbrook}, \bibinfo{title}{Differentiating the pseudo
  determinant}, \bibinfo{journal}{Linear Algebra and its Applications}
  \bibinfo{volume}{548} (\bibinfo{year}{2018}) \bibinfo{pages}{293--304}.
\bibitem[{Holbrook(2021)}]{holbrook2021quantum}
\bibinfo{author}{A.~J. Holbrook}, \bibinfo{title}{A quantum parallel markov
  chain monte carlo}, \bibinfo{journal}{arXiv preprint arXiv:2112.00212}
  (\bibinfo{year}{2021}).
\bibitem[{Ihaka and Gentleman(1996)}]{ihaka1996r}
\bibinfo{author}{R.~Ihaka}, \bibinfo{author}{R.~Gentleman},
  \bibinfo{title}{{R}: a language for data analysis and graphics},
  \bibinfo{journal}{{J}ournal of {C}omputational and {G}raphical {S}tatistics}
  \bibinfo{volume}{5} (\bibinfo{year}{1996}) \bibinfo{pages}{299--314}.
\bibitem[{Lee(2012)}]{Lee2012}
\bibinfo{author}{J.~M. Lee}, \bibinfo{title}{Quotient {M}anifolds},
  \bibinfo{publisher}{Springer New York}, \bibinfo{address}{New York, NY},
  \bibinfo{year}{2012}, pp. \bibinfo{pages}{540--563}.
\bibitem[{Li et~al.(2013)Li, Ranka and Sahni}]{li2013gpu}
\bibinfo{author}{J.~Li}, \bibinfo{author}{S.~Ranka},
  \bibinfo{author}{S.~Sahni}, \bibinfo{title}{Gpu matrix multiplication},
  \bibinfo{journal}{Multicore Computing: Algorithms, Architectures, and
  Applications} \bibinfo{volume}{345} (\bibinfo{year}{2013}).
\bibitem[{Liu et~al.(2000)Liu, Liang and Wong}]{liu2000multiple}
\bibinfo{author}{J.~S. Liu}, \bibinfo{author}{F.~Liang}, \bibinfo{author}{W.~H.
  Wong}, \bibinfo{title}{The multiple-try method and local optimization in
  metropolis sampling}, \bibinfo{journal}{Journal of the American Statistical
  Association} \bibinfo{volume}{95} (\bibinfo{year}{2000})
  \bibinfo{pages}{121--134}.
\bibitem[{Luo and Tjelmeland(2019)}]{luo2019multiple}
\bibinfo{author}{X.~Luo}, \bibinfo{author}{H.~Tjelmeland}, \bibinfo{title}{A
  multiple-try metropolis--hastings algorithm with tailored proposals},
  \bibinfo{journal}{Computational Statistics} \bibinfo{volume}{34}
  (\bibinfo{year}{2019}) \bibinfo{pages}{1109--1133}.
\bibitem[{Metropolis et~al.(1953)Metropolis, Rosenbluth, Rosenbluth, Teller and
  Teller}]{metropolis1953equation}
\bibinfo{author}{N.~Metropolis}, \bibinfo{author}{A.~W. Rosenbluth},
  \bibinfo{author}{M.~N. Rosenbluth}, \bibinfo{author}{A.~H. Teller},
  \bibinfo{author}{E.~Teller}, \bibinfo{title}{Equation of state calculations
  by fast computing machines}, \bibinfo{journal}{The {J}ournal of {C}hemical
  {P}hysics} \bibinfo{volume}{21} (\bibinfo{year}{1953})
  \bibinfo{pages}{1087--1092}.
\bibitem[{Neal(2003)}]{neal2003slice}
\bibinfo{author}{R.~M. Neal}, \bibinfo{title}{{Slice sampling}},
  \bibinfo{journal}{The Annals of Statistics} \bibinfo{volume}{31}
  (\bibinfo{year}{2003}) \bibinfo{pages}{705 -- 767}.
\bibitem[{Peskun(1973)}]{peskun1973optimum}
\bibinfo{author}{P.~H. Peskun}, \bibinfo{title}{Optimum monte-carlo sampling
  using markov chains}, \bibinfo{journal}{Biometrika} \bibinfo{volume}{60}
  (\bibinfo{year}{1973}) \bibinfo{pages}{607--612}.
\bibitem[{Rosenthal(2011)}]{rosenthal2011optimal}
\bibinfo{author}{J.~S. Rosenthal}, \bibinfo{title}{Optimal proposal
  distributions and adaptive mcmc}, \bibinfo{journal}{Handbook of Markov Chain
  Monte Carlo}  (\bibinfo{year}{2011}) \bibinfo{pages}{119--138}.
\bibitem[{Salmon et~al.(2011)Salmon, Moraes, Dror and
  Shaw}]{salmon2011parallel}
\bibinfo{author}{J.~K. Salmon}, \bibinfo{author}{M.~A. Moraes},
  \bibinfo{author}{R.~O. Dror}, \bibinfo{author}{D.~E. Shaw},
  \bibinfo{title}{Parallel random numbers: As easy as 1, 2, 3}, in:
  \bibinfo{booktitle}{Proceedings of 2011 International Conference for High
  Performance Computing, Networking, Storage and Analysis}, SC '11,
  \bibinfo{publisher}{Association for Computing Machinery},
  \bibinfo{address}{New York, NY, USA}, \bibinfo{year}{2011}.
\bibitem[{Schwedes and Calderhead(2021)}]{schwedes2021rao}
\bibinfo{author}{T.~Schwedes}, \bibinfo{author}{B.~Calderhead},
  \bibinfo{title}{Rao-{B}lackwellised parallel {MCMC}}, in:
  \bibinfo{editor}{A.~Banerjee}, \bibinfo{editor}{K.~Fukumizu} (Eds.),
  \bibinfo{booktitle}{Proceedings of The 24th International Conference on
  Artificial Intelligence and Statistics}, volume \bibinfo{volume}{130} of
  \text{\bibinfo{series}{Proceedings of Machine Learning Research}},
  \bibinfo{publisher}{PMLR}, \bibinfo{year}{2021}, pp.
  \bibinfo{pages}{3448--3456}.
\bibitem[{Stewart(1980)}]{stewart1980efficient}
\bibinfo{author}{G.~W. Stewart}, \bibinfo{title}{The efficient generation of
  random orthogonal matrices with an application to condition estimators},
  \bibinfo{journal}{SIAM Journal on Numerical Analysis} \bibinfo{volume}{17}
  (\bibinfo{year}{1980}) \bibinfo{pages}{403--409}.
\bibitem[{Tierney(1994)}]{tierney1994markov}
\bibinfo{author}{L.~Tierney}, \bibinfo{title}{Markov chains for exploring
  posterior distributions}, \bibinfo{journal}{The Annals of Statistics}
  \bibinfo{volume}{22} (\bibinfo{year}{1994}) \bibinfo{pages}{1701--1728}.
\bibitem[{Tjelmeland(2004)}]{tjelmeland2004using}
\bibinfo{author}{H.~Tjelmeland}, \bibinfo{title}{Using All Metropolis-Hastings
  Proposals to Estimate Mean Values (Norwegian University of Science and
  Technology, Trondheim, Norway)}, \bibinfo{type}{Technical Report}, Tech Rep
  4, \bibinfo{year}{2004}.
\bibitem[{Williams and Barber(1998)}]{williams1998bayesian}
\bibinfo{author}{C.~K. Williams}, \bibinfo{author}{D.~Barber},
  \bibinfo{title}{Bayesian classification with {G}aussian processes},
  \bibinfo{journal}{IEEE Transactions on Pattern Analysis and Machine
  Intelligence} \bibinfo{volume}{20} (\bibinfo{year}{1998})
  \bibinfo{pages}{1342--1351}.
\bibitem[{Yang et~al.(2018)Yang, Chen, Bernton and
  Liu}]{yang2018parallelizable}
\bibinfo{author}{S.~Yang}, \bibinfo{author}{Y.~Chen},
  \bibinfo{author}{E.~Bernton}, \bibinfo{author}{J.~S. Liu}, \bibinfo{title}{On
  parallelizable {M}arkov chain {M}onte {C}arlo algorithms with
  waste-recycling}, \bibinfo{journal}{Statistics and Computing}
  \bibinfo{volume}{28} (\bibinfo{year}{2018}) \bibinfo{pages}{1073--1081}.
\bibitem[{Yang and Liu(2021)}]{yang2021convergence}
\bibinfo{author}{X.~Yang}, \bibinfo{author}{J.~S. Liu},
  \bibinfo{title}{Convergence rate of multiple-try metropolis independent
  sampler}, \bibinfo{journal}{arXiv preprint arXiv:2111.15084}
  (\bibinfo{year}{2021}).

\end{thebibliography}

%
%

%
%


\end{document}